\documentclass[sigconf, nonacm]{acmart}

\newcommand\vldbdoi{XX.XX/XXX.XX}
\newcommand\vldbpages{XXX-XXX}
\newcommand\vldbvolume{14}
\newcommand\vldbissue{1}
\newcommand\vldbyear{2020}
\newcommand\vldbauthors{\authors}
\newcommand\vldbtitle{\shorttitle} 
\newcommand\vldbavailabilityurl{https://github.com/johnmzz/temproal_core_CC}
\newcommand\vldbpagestyle{plain}

\usepackage{hyperref}
\usepackage{graphicx,mathtools,kantlipsum}
\usepackage[normalem]{ulem}
\usepackage{amsmath,amsfonts}
\usepackage{listings}
\usepackage{color}
\usepackage{courier}
\usepackage{float}
\usepackage{makecell}
\usepackage[noend]{algpseudocode}
\usepackage{todonotes}
\usepackage{booktabs}
\usepackage[capitalize]{cleveref}
\usepackage{balance}
\usepackage{bm}
\usepackage{multicol}
\usepackage{multirow}
\usepackage{xspace}
\usepackage{tabularx}

\usepackage[ruled, vlined, linesnumbered]{algorithm2e}
\usepackage{caption}
\usepackage{subcaption}
\usepackage{url}
\usepackage{enumitem}

\usepackage{lipsum}
\usepackage{epsfig}

\setlist[itemize]{nosep,label=\textbf{-},leftmargin=*}
\setlist[enumerate]{nosep,label=\textit{(\arabic*)},leftmargin=*}

\newcommand{\kw}[1]{{\ensuremath {\mathsf{#1}}}\xspace}

\newcommand{\reffig}[1]{Figure~\ref{fig:#1}}

\newcommand{\refsubsec}[1]{Section~\ref{subsec:#1}}
\newcommand{\reftab}[1]{Table~\ref{tab:#1}}
\newcommand{\refalg}[1]{Algorithm~\ref{alg:#1}}
\newcommand{\refdef}[1]{Definition~\ref{def:#1}}

\newcommand{\reflem}[1]{Lemma~\ref{lem:#1}}

\newcommand{\stitle}[1]{\noindent{\bf #1}}

\newcommand{\kwnull}{\ensuremath{\kw{Null}}}
\newcommand{\kwcontinue}{\textbf{continue}}

\newcommand{\algquery}{\kw{Query}}

\newcommand{\idx}{PECB-Index\xspace}
\newcommand{\algmerge}{\kw{Merge}}

\newcommand{\binforest}{\mathcal{B}}
\newcommand{\la}{\langle}
\newcommand{\ra}{\rangle}

\newcommand{\eeq}{EC\xspace}

\newcommand{\ct}{\ensuremath{\mathcal{CT}}}

\newcommand{\msf}{\mathcal{M}}

\newcommand{\algfind}{\kw{findInsertion}}
\newcommand{\algconstruct}{\kw{\binforest}-Construct}
\newcommand{\tcore}{\mathcal{T}^k_{[ts,te]}}

\SetKwProg{proc}{Procedure}{}{}

\begin{document}
\title{Accelerating Historical K-Core Search in Temporal Graphs}

%%
%% The "author" command and its associated commands are used to define the authors and their affiliations.
\author{Zhuo Ma}
\affiliation{%
  \institution{The University of New South Wales}
  \city{Sydney}
  \country{Australia}
}
\email{zhuo.ma@student.unsw.edu.au}

\author{Dong Wen}
\affiliation{%
  \institution{The University of New South Wales}
  \city{Sydney}
  \country{Australia}
}
\email{dong.wen@unsw.edu.au}

\author{Kaiyu Chen}
\affiliation{%
  \institution{The University of New South Wales}
  \city{Sydney}
  \country{Australia}
}
\email{kaiyu.chen1@unsw.edu.au}

\author{Yixiang Fang}
\affiliation{%
  \institution{The Chinese University of Hong Kong}
  \city{Shenzhen}
  \country{China}
}
\email{fangyixiang@cuhk.edu.cn}

\author{Xuemin Lin}
\affiliation{%
  \institution{Shanghai Jiao Tong University}
  \city{Shanghai}
  \country{China}
}
\email{xuemin.lin@sjtu.edu.cn}

\author{Wenjie Zhang}
\affiliation{%
  \institution{The University of New South Wales}
  \city{Sydney}
  \country{Australia}
}
\email{wenjie.zhang@unsw.edu.au}

\begin{abstract}
% The connected component (CC) is a fundamental concept in graph analysis and serves as the basis for many applications on temporal graphs, where edges are timestamped and connectivity evolves over time. Analyzing CCs over historical windows is critical for tasks such as contact tracing, fault diagnosis, and financial forensics. While existing methods support historical CC and connectivity queries, they do not efficiently retrieve the connected component containing a specific vertex. We formalize this problem as historical connected component search and propose an efficient index-based solution. 
We study the temporal $k$-core component search (TCCS), which outputs the $k$-core containing the query vertex in the snapshot over an arbitrary query time window in a temporal graph. The problem has been shown critical for tasks such as contact tracing, fault diagnosis, and financial forensics. The state-of-the-art EF-Index designs a separated forest structure for a set of carefully selected windows, incurring quadratic preprocessing time and large redundant storage. Our method introduces the ECB-forest, a compact edge-centric binary forest that captures $k$-core of any arbitrary query vertex over time. In this way, query can be processed by searching a connected component in the forest. We develop an efficient algorithm for index construction. Experiments on real-world temporal graphs show that our method significantly improves the index size and construction cost (up to 100x faster on average) while maintaining the high query efficiency.
\end{abstract}

\maketitle

\pagestyle{\vldbpagestyle}
\begingroup\small\noindent\raggedright\textbf{PVLDB Reference Format:}\\
\vldbauthors. \vldbtitle. PVLDB, \vldbvolume(\vldbissue): \vldbpages, \vldbyear.\\
\href{https://doi.org/\vldbdoi}{doi:\vldbdoi}
\endgroup
\begingroup
\renewcommand\thefootnote{}\footnote{\noindent
This work is licensed under the Creative Commons BY-NC-ND 4.0 International License. Visit \url{https://creativecommons.org/licenses/by-nc-nd/4.0/} to view a copy of this license. For any use beyond those covered by this license, obtain permission by emailing \href{mailto:info@vldb.org}{info@vldb.org}. Copyright is held by the owner/author(s). Publication rights licensed to the VLDB Endowment. \\
\raggedright Proceedings of the VLDB Endowment, Vol. \vldbvolume, No. \vldbissue\ %
ISSN 2150-8097. \\
\href{https://doi.org/\vldbdoi}{doi:\vldbdoi} \\
}\addtocounter{footnote}{-1}\endgroup

\ifdefempty{\vldbavailabilityurl}{}{
\vspace{.3cm}
\begingroup\small\noindent\raggedright\textbf{PVLDB Artifact Availability:}\\
The source code, data, and/or other artifacts have been made available at \url{\vldbavailabilityurl}.
\endgroup
}

% \xx{todo: 1. American/British style: keep consistent, 2. double check the start/end time (the previous version anchors the end time, now change to anchoring start time. There may exists some error. }

% \rr{done:3. temporal k-core connected component -> temporal k-core component)}

%!TEX root = main.tex
\section{Introduction}
\label{sec:intro}
Temporal graphs are widely used to model evolving interactions between entities, where each edge is associated with a timestamp indicating the time of interaction. These graphs arise in real-world systems, such as infrastructural networks, communication networks, and biological networks \cite{holme2012temporal}. Figure \ref{fig:graph} illustrates an example of a temporal graph, where each edge $(u, v, t)$ represents an interaction between vertices $u$ and $v$ at time $t$.

\begin{figure}[htbp]
    \centering
    \begin{subfigure}{0.49\columnwidth}
        \centering
        \includegraphics[width=\linewidth]{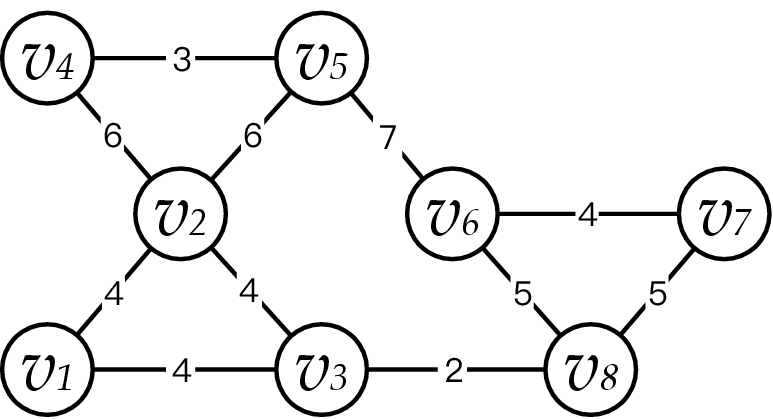}
        \caption{$G$}
        \label{fig:graph}
    \end{subfigure}
    \begin{subfigure}{0.49\columnwidth}
        \centering
        \includegraphics[width=\linewidth]{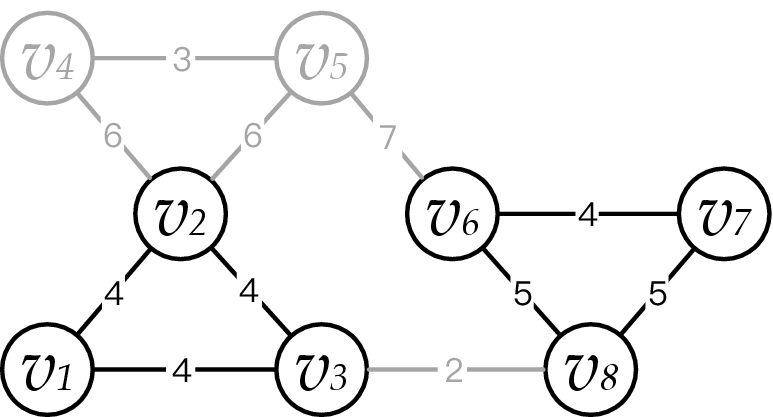}
        \caption{$G_{[4,5]}$}
        \label{fig:projected}
    \end{subfigure}
    \vspace{-0.5em}
    \caption{A temporal graph $G$ and its projected graph $G_{[4,5]}$.}
    \label{fig:temporal_graph}
    
\end{figure}

% The $k$-core \cite{Seidman1983} of a graph is the \rr{maximal induced connected subgraph}
The $k$-cores \cite{Seidman1983} of a graph are the maximal induced connected subgraphs in which every vertex has degree at least $k$ within the subgraph. Thanks to its simplicity and effectiveness at capturing cohesion, the $k$-core has attracted significant research attention due to its wide range of applications, including community detection, network visualization, and system structure analysis \cite{Cui2014, Li2018, Zhang2010, Cheng2011}. Detecting $k$-cores in temporal graphs supports applications such as uncovering suspicious account clusters during anti–money laundering investigations over specified time intervals \cite{Chu2019, lin2024a}.

\stitle{Temporal $k$-core component search (TCCS).}
Given any time window in a temporal graph, a snapshot is induced by all edges falling in the window. Consequently, the $k$-core structure varies across snapshots of different time windows. Great attention has been drawn in recent studies to efficiently retrieve results for an arbitrary window.
Yu et al. \cite{DBLP:journals/pvldb/YuWQ00021} proposed an index to identify whether a query vertex is contained in the $k$-core of a query time window. Yang et al. \cite{yang2023} enumerated all distinct $k$-cores in all sub-windows of the query time window. Both works focus on global analysis and lack support for efficiently retrieving the component containing a query vertex, which is crucial for monitoring a user’s local community \cite{DBLP:conf/kdd/TangLZN08} or delineating a fault region \cite{DBLP:journals/ress/LevitinXY14}. To address this, Yang et al. \cite{yang2024} investigated the $k$-core search problem (namely TCCS for short), which identifies all vertices that are connected to the query vertex in the $k$-core of the querying window.

\stitle{Applications.}
By unifying cohesiveness and reachability in a time window, TCCS serves as a versatile primitive in several time-critical domains and accelerates diverse analyses while preserving the exact structural context in which critical events unfold.

\begin{itemize}
    \item \textit{Financial anomaly detection.} In financial crime analysis, transaction networks are scanned for dense clusters that flare up around major events; by returning only the $k$-core component that contains a black-listed account, TCCS cuts the number of accounts an investigator must review from hundreds to a handful, streamlining anti-money-laundering workflows \cite{yang2023}.
    
    \item \textit{Social Network Analysis.} On social platforms, users that persist in temporal cores exhibit a strong propensity to spread misinformation. Isolating their surrounding component helps moderators focus on the most influential part of a rumour cascade \cite{momin2023}.
    
    \item \textit{Epidemiology and Security.} In contact or communication networks, given a confirmed infection or a suspect entity, TCCS can identify all individuals potentially exposed or connected within a time window. This supports contact tracing \cite{mao2021digital}, exposure risk assessment \cite{DBLP:journals/pacmmod/XieFXLM23}, and criminal associate detection \cite{DBLP:conf/www/AnyanwuS03}. 
\end{itemize}

\stitle{The STOA.} Yang et al. \cite{yang2024} proposed an index-based solution for the TCCS problem. The proposed index, namely EF-Index, leverages Temporal $k$-Core Evolution, modeling structural changes in $k$-cores over time windows $[ts, te]$ in temporal graphs. Two thresholds define these transitions: the Tightest Time Interval (TTI), the smallest window where a $k$-core’s topology remains identical, and the Loosest Time Interval (LTI), the broadest such window. Expanding the LTI or shrinking the TTI triggers evolution into a new subgraph. The Core Lineage captures evolutionary sequences where each $k$-core evolves from its predecessor as the time window expands, forming a Lineage Graph—a directed acyclic graph with nodes as distinct $k$-cores and edges as lineage links. The EF-Index stores connectivity in Minimum Temporal Spanning Forests (MTSFs), incrementally tracking vertex connectivity across chains. For TCCS queries, EF-Index uses a lookup structure to map time windows to the relevant $k$-core and MTSF, followed by a label-constrained depth-first search (DFS) on the MTSF, achieving query complexity $O(d*log(p_{max}) + \sum_v deg(v))$, where $d$ is the depth of the lineage graph and $p_{max}$ is the maximum number of LTIs of a temporal $k$-core. The EF-Index faces several limitations:

\begin{itemize}
    \item High Computational Cost: OTCD computes all temporal $k$-cores, including redundant ones with identical components, yielding a costly $O(t_{\max}^2 \cdot V_k)$ complexity, where $V_k$ is the average number of vertices for all temporal $k$-cores.

    \item Storage Redundancy: MTSFs are stored for each chain, despite overlapping components across $k$-cores, inflating space usage.

    % \item Maintenance Overhead: Constructing and updating the Lineage Graph and chain cover add complexity proportional to the number of $k$-cores.
\end{itemize}

% These issues underscore the need for a more efficient approach to reduce redundancy and overhead while preserving query performance.

\stitle{Our Approach}  
The drawbacks of the EF-Index motivates us to develop a more efficient approach to reduce redundancy and overhead while preserving query performance.
Our basic idea is to reduce the TCCS to a subproblem called start-anchored TCCS, where the start time $ts$ is fixed, and queries allow arbitrary end times $te \ge ts$. 
% We design a forest structure (storing the parent and the children of each forest node) for the start-anchored TCCS. For the whole index (for arbitrary windows), we only store neighbors for a forest node at a start time if they are different at the previous start time.

% To support this efficiently, we define the concept of EC-equivalence: a set of temporal edges is EC-equivalent to the original temporal $k$-core with respect to $ts$ if it preserves all connected components for any $te \ge ts$.
Given the anchored start time, we utilize a concept for any edge $e$ called core time, which is the earliest end time such that $e$ is in the $k$-core of the time window.
We observe that the minimum spanning forest (MSF) of the temporal graph by the core time of each edge as weight is equivalent to the original graph in terms of the start-anchored TCCS problem. By equivalence, we mean the result can be derived by collecting the connected component of the query vertex in the window of the anchored start time and an arbitrary end time in the forest.
However, simply indexing as MSF introduces much space redundancy when storing the structure for different start times. The large index also increases the search space and slows down query performance. To address this, we design a new structure called the ECB-forest (Edge-centric, Connectivity-equivalent, Binary forest), where each node corresponds to a temporal edge in MSF and maintains at most two children. Similar as the MSF, the structure is proved still equivalent to the graph, and the size is bounded by $O(n)$, where $n$ is the number of vertices in the graph.
The binary property reduces redundant traversal paths and improves query efficiency without compromising the equivalence property. 

% For each start time $ts$, the ECB-forest captures edge-centric connectivity in this binary form. 
We organize all ECB-forests across start times into a unified index, denoted \idx. Each edge maintains a list of versioned entries, each storing the parent and children at a specific $ts$. For the whole index (for arbitrary windows), we only store neighbors for a forest node at a start time if they are different at the previous start time. In this way, neighbors of a node in the ECB-forest for any start time can be derived by a binary search.

To answer a query $(u, [ts, te])$, we locate all temporal edges involving $u$ within the temporal $k$-core of $[ts, te]$, and use them as entry points in the ECB-forest at $ts$. The neighbors of each forest node at corresponding start time can be derived by a binary search. The search explores both upward and downward directions, following parent and child links whose core times fall within the query window, and collects all connected edges—and hence, all connected vertices.

We construct the index in an incremental way by iterating the start time. Starting from an empty forest, each new edge is integrated into the ECB-forest to get the result for the next start time. For adding each edge (forest node) $e$ to the forest, we first carefully pick connect $e$ to certain existing nodes to produce a graph (a cycle may exist) that are equivalent to the graph in terms of query processing. Then we develop a set of constant-time transformation operators to eliminate the cycle and return a valid ECB-Forest. As a result, the running time for processing each edge is bounded by the depth of the forest.

\stitle{Contributions.}
We summarize our main contributions as follows:
\begin{itemize}
    \item We propose a novel edge-centric binary forest structure for the TCCS problem. The structure enables efficient $k$-core search for arbitrary vertices and windows.
    \item We propose an efficient algorithm for index construction. Our algorithm bound the cost to add each edge to the forest by the depth of the forest (i.e., the longest shortest path distance for any forest node to the root).
    \item We conduct extensive experiments on real-world temporal graphs. The results demonstrate our significant improvement of efficiency and space usage compared with the state-of-the-art solution.
\end{itemize}

%!TEX root = main.tex
\section{Preliminary}
\label{sec:pre}

We consider an undirected temporal graph $G(V,E)$, where each edge $(u,v,t) \in E$ is associated with a timestamp $t$, representing the interaction time between vertices $u$ and $v$. Let $n$ and $m$ denote the number of vertices and edges in $G$, respectively. Without loss of generality, we assume that edge timestamps form a continuous sequence of integers starting from 1. The maximum number of distinct timestamps in the graph is denoted by $t_{max}$. The timestamp of edge $e$ is denoted by $t(e)$. The degree of a vertex $u$ is represented as $deg(u)$. We use $E_t$ to denote the set of edges with timestamp $t$. The projected graph of $G$ over a time window $[ts,te]$, denoted as $G_{[ts,te]}$, is the subgraph consisting of all edges whose timestamps fall within $[ts,te]$.

\begin{definition}[$K$-Core \cite{Seidman1983}]
\label{def:k-core}
Given a simple graph $G$ and an integer $k$, a $k$-core of $G$ is a maximal induced connected subgraph of $G$ in which every vertex has at least $k$ neighbors.
\end{definition}

% We use $\core_k(G)$ to denote the $k$-core of a graph $G$. The concept of a temporal $k$-core extends the idea of a $k$-core to temporal graphs, as defined below.

\begin{definition}[Temporal $K$-Core \cite{yang2023}]
\label{def:temporal-core}
Given a temporal graph $G$ and a time window $[ts,te]$, the temporal $k$-core of $[ts,te]$, denoted as $\tcore$, is the maximal subgraph of $G_{[ts,te]}$ in which every vertex has at least $k$ neighbors.
\end{definition}

We omit the superscript $k$ in $\tcore$ for simplicity when it is clear from the context.
Note that the connectivity property is ignored in the definition of temporal $k$-core, which exists in the initial definition \cite{Seidman1983}. For clearance, we follow the \refdef{temporal-core}, and we use the \textit{temporal $k$-core component} to represent a connected component in the temporal $k$-core. We present the research problem \cite{yang2024} as follows.

\stitle{Problem Statement.} Given a temporal graph $G$, integer $k$, a time window $[ts,te]$, and a vertex $u$, we aim to retrieve all vertices in the temporal $k$-core component containing $u$.

% our technique works for any interger $k$. We mainly focus on a specific k in the rest for simplicity.
% Given an edge set $E$, we use $E_{\ge t}$ to denote all edges in $E$ with time labels not earlier than $t$.

\vspace{-0.5em}
\begin{example}
\reffig{projected} shows the projected graph of $G$ induced by all edges whose timestamps lie in the interval $[4,5]$. In this window we retain exactly six vertices, namely $\{v_1, v_2, v_3, v_6, v_7, v_8\}$, and six edges, while all other edges fall outside the window. Since every vertex in each triangle has degree 2, there are two temporal 2‐core components of $\mathcal{T}_{[4,5]}$, namely $\{v_1,v_2,v_3\}$ and $\{v_6,v_7,v_8\}$.
\end{example}

%!TEX root = main.tex
\section{Existing Studies}
\label{sec:base}

\subsection{The State of the Art}

An index structure called Evolution Forest Index (EF-Index) is proposed in \cite{yang2024} Temporal $k$-Core Component Search (TCCS) queries. The EF-Index builds upon the concept of Temporal $k$-Core Evolution, which describes the dynamic changes in $k$-cores within temporal graphs as they evolve over varying time windows $[ts,te]$. A temporal $k$-core evolves by expanding or shrinking its structure as the time window varies, reflecting the evolution of the underlying projected graph $G_{[ts,te]}$. Two critical thresholds, the Loosest Time Interval (LTI) and the Tightest Time Interval (TTI), characterize these transitions. A TTI is defined as the smallest time window within which a temporal $k$-core remains topologically identical. Conversely, the LTI is defined as the broadest time window over which a temporal $k$-core retains its structure. Expanding the LTI or shrinking the TTI forces the temporal $k$-core to evolve into a new subgraph, capturing the temporal dynamics of the graph. Notably, each temporal $k$-core has a unique TTI but may have multiple LTIs.

% \begin{definition}[Tightest Time Interval]
% \label{def:tti}
% For a temporal $k$-core $\tcore$, a time window $[ts',te']$ is the Tightest Time Interval (TTI) of $\tcore$ if and only if:
% \begin{enumerate}
%     \item $T_k[ts',te']$ is topologically identical to $\tcore$, and 
%     \item there does not exist another time window $[ts'',te''] \subset [ts',te']$ such that $T_k[ts'',te'']$ also topologically identical to $\tcore$.
% \end{enumerate}
% \end{definition}

% \begin{definition}[Loosest Time Interval]
% \label{def:lti}
% For a temporal $k$-core $\tcore$, a time window $[ts',te']$ is a Loosest Time Interval (LTI) of $\tcore$ if and only if:
% \begin{enumerate}
%     \item $T_k[ts',te']$ is topologically identical to $\tcore$, and 
%     \item there does not exist another time window $[ts'',te''] \supset [ts',te']$ such that $T_k[ts'',te'']$ also topologically identical to $\tcore$.
% \end{enumerate}
% \end{definition}

The concept of Core Lineage was introduced to capture the evolutionary relationships between temporal $k$-cores. A core lineage represents a sequence of temporal $k$-cores where each core evolves into the next as the time window expands. Importantly, each temporal $k$-core in a lineage contains its predecessor, capturing a nested structure. However, this evolution is not strictly linear. A single $k$-core may evolve into multiple $k$-cores, and multiple $k$-cores may converge into one. These relationships are represented as a Lineage Graph, a directed acyclic graph where nodes correspond to distinct temporal $k$-cores and edges capture their lineage relationships.

The EF-Index is constructed based on the Lineage Graph, which provides a compact representation of the evolutionary relationships among temporal $k$-cores. The construction begins with computing all distinct $k$-cores for all possible time windows using the OTCD algorithm \cite{yang2023}. This step incurs significant computational costs, with a time complexity of $O(t_{max}^2 * V_k)$, where $V_k$ is the average number of vertices in the temporal $k$-cores. Once the $k$-cores are computed, the lineage relationships are identified using Corner Time Intervals (CTIs), derived from the LTIs and TTIs of each temporal $k$-core. These CTIs facilitate efficient construction of the Lineage Graph. To reduce redundancy, a lineage chain cover is generated using the Hopcroft-Karp algorithm. This cover minimizes the number of chains while ensuring that each $k$-core appears exactly once across the chains. Furthermore, only the largest $k$-cores at the tail of each lineage chain are preserved, significantly reducing storage requirements.

To compress connectivity information within temporal $k$-cores, the EF-Index employs Minimum Temporal Spanning Forests (MTSFs). An MTSF incrementally maintains vertex connectivity across lineage chains. Each edge in the MTSF is labeled with a valid time window $[ts,te]$, indicating when its endpoints are connected. MTSFs are constructed by comparing consecutive $k$-cores in a lineage chain and adding edges that connect previously disconnected vertices. This approach achieves a time complexity of $O(E_k + log(V_k))$ based on its implementation, where $E_k$ is the average number of edges in the temporal $k$-cores.

Finally, to process TCCS queries, the EF-Index employs a lookup structure that maps query time windows to the corresponding $k$-core and its lineage chain. Each entry in this lookup structure contains the $k$-core’s TTI as its identifier, a pointer to the corresponding MTSF, and a list of pointers to its lineage $k$-cores. Query processing involves two steps. First, the relevant MTSF is retrieved by traversing the lookup structure to find the $k$-core with the largest TTI containing the query time window. This retrieval step has a complexity of $O(d * log(p_{max}))$, where $d$ is the depth of the lineage graph and $p_{max}$ is the maximum number of LTIs of a temporal $k$-core. Second, a label-constrained depth-first search (DFS) is performed on the retrieved MTSF, restricted to edges whose labels are subintervals of the queried time window. The DFS step has a complexity of $O(\Sigma_{v} deg(v))$.

\stitle{Challenges.} 
While the EF-Index provides a robust framework for answering Temporal $k$-Core Component Search (TCCS) queries, it suffers from several notable limitations. First, the OTCD algorithm computes all temporal $k$-cores by treating each distinct subgraph (edge-set) as a separate result. However, different edge-sets may yield the same connected components, resulting in unnecessary computations and increased redundancy. The time complexity of OTCD, $O(t_{max}^2 * V_k)$, further highlights its computational overhead. Second, constructing the lineage graph and finding the optimal lineage cover incur additional costs proportional to the number of temporal $k$-cores and their average Corner Time Intervals (CTIs), adding further to the computational burden. Third, while the Minimum Temporal Spanning Forests (MTSFs) are stored incrementally, an MTSF is created for every $k$-core chain, even though the connected components of different $k$-cores in a chain may overlap. This redundancy leads to a bloated index size. 
% Finally, querying the EF-Index involves first identifying the relevant MTSF by traversing the lookup structure, which has a complexity of $O(d * log(p_{max}))$. In the worst case, $d$ scales with $t_{max}$, adversely affecting query efficiency. 
Overall, the EF-Index’s construction process is computationally expensive, its storage requirements are excessive due to redundancy.
% and its query performance is limited by the inefficiency of the initial MTSF retrieval step.

\vspace{-1em}
\subsection{Other Related Works}

Computing the $k$-core of a static graph is a well-established problem. The standard approach solves it in linear time $O(n + m)$ by iteratively removing vertices with the smallest degree~\cite{batagelj2003m, Jason2005}. To scale to massive graphs, optimized solutions have been developed under external~\cite{Cheng2011}, semi-external~\cite{khaouid2015k, Wen2015}, parallel~\cite{dhulipala2017julienne}, and distributed settings~\cite{montresor2011distributed, pechlivanidou2014mapreduce}. For dynamic graphs, where edges are incrementally inserted or deleted, several techniques enable efficient $k$-core maintenance without full recomputation~\cite{Sariyuce2013, Li2014, Zhang2016}. Further extensions adapt these methods to streaming and distributed environments using localized updates and vertex reordering strategies~\cite{aksu2014distributed, Wen2015}.

Research on cohesive patterns in temporal graphs has branched into many variants, each adding a different time-aware constraint. Historical $k$-cores capture the structure at designated snapshots \cite{DBLP:journals/pvldb/YuWQ00021}, whereas span-cores demand that every edge persist throughout a window \cite{Galimberti2018}. Other models impose persistence thresholds $(\pi, \tau)$ \cite{Li2018}, combine structural and temporal filters around a query vertex \cite{Li2021}, or weight edges by interaction frequency \cite{Ma2020}. Further extensions include the $(k, h)$-core, which requires each vertex to have at least $k$ neighbors with $h$ interactions \cite{Wu2015}; density-bursting subgraphs that maximize growth rate \cite{Chu2019}; periodic $k$-cores that recur regularly \cite{Qin2019,Qin2022}; and quasi-$(k, h)$-cores, which refine these ideas for efficient incremental maintenance \cite{Bai2020}. Together, these models deepen our understanding of cohesive substructures in evolving networks.

% Beyond connectivity, scholars have explored shortest paths \cite{wu2014}, span-reachability \cite{DBLP:conf/icde/WenHZQZ020}, temporal motifs \cite{gurukar2015}, and a range of cohesive substructures—including k-cores \cite{Li2018, Wu2015, DBLP:journals/pvldb/YuWQ00021}, quasi-cliques \cite{yang2016}, and dense subgraphs \cite{ma2017} in temporal data. Unlike those studies, the present work is, to our knowledge, the first to formalise window-CC and window-SCC queries—components defined over arbitrary time windows without requiring time-respecting paths—and to develop dedicated index-based solutions with linear-space guarantees.

%!TEX root = main.tex
\section{Our Index Structure}
\label{sec:idx}

Designing an index-based algorithm for TCCS queries presents several key challenges. First, the index must preserve the coreness of vertices, ensuring that it accurately represents the sets of vertices belonging to the temporal $k$-core of any given time window $[ts,te]$. This involves capturing and storing the cohesiveness properties of the graph across varying temporal dimensions. Second, the index must also preserve the connectivity among these vertices within the specified time window, enabling efficient identification of connected components. Maintaining this dual representation of both coreness and connectivity is non-trivial, as it requires balancing comprehensiveness with storage efficiency. 

% \subsection{The Framework}
% \label{subsec:frk}
% Our indexing idea is inspired by the POEC-Index \cite{DBLP:journals/pacmmod/SongWXQZ024}, which computes whether two query vertices are connected via a path in a query window. We first introduce the framework and then clarify the considerable difference between our method and theirs. 

\stitle{Our Framework.}
\label{subsec:frk}
We mainly focus on the index for a specific $k$, and our technique can be naturally extended for any possible integer $k$. We start by considering a sub-problem as follows.

\begin{definition}[Start-Anchored Temporal K-Core Search]
\label{def:sub}
Given a temporal graph $G$, a predefined fixed start time $ts$, a vertex $u$, and an arbitrary query end time $te$, the start-anchored $k$-core search retrieves all vertices in the temporal $k$-core containing $u$ in the projected graph of $G$ over $[ts,te]$.
\end{definition}

We will develop a forest structure for the sub-problem where a fixed start time 
$ts$ are given. Given a query vertex and an arbitrary end time, the answer is obtained by traversing the forest instead of the whole graph. Each forest node stores pointers to its parent and children to support the search. 
Now, considering all possible start times (i.e., our research problem), we observe only a small proportion of nodes update their neighbors when comparing the forests for two adjacent start times. Leveraging this observation, we iterate start times in descending order and record a node’s neighborhood (in the forest) only when it differs from what was stored at the start time of the previous round. In other words, we compress the forests for all start times in the index, and a binary search is conducted to derive the neighbors of a node for a specific start time. 

\subsection{The Index by Anchoring Start Time}
\label{subsec:index}

In this subsection, we propose the index designed for start-anchored temporal $k$-core search. In addition to determining whether a vertex lies in the $k$-core of the query snapshot, the index must also preserve the connectivity among all vertices within that $k$-core.
Our key idea is to construct a compact structure that preserves, for every temporal $k$-core, exactly the same connectivity relationships as the original temporal $k$-core. The equivalence property is formally defined as follows.

\begin{definition}[\eeq-Equivalence]
Given a temporal graph $G$ and a start time $ts$, a set of temporal edges $F$ is \eeq-equivalent to the temporal $k$-core $\mathcal{T}$ w.r.t. $ts$ if for any end time $te$ ($ts \le te$), the connected components in the projected induced sub-forest of $F$ over $[ts,te]$ are the same as those of $\mathcal{T}$.
\end{definition}

We simply say $F$ is \eeq-equivalent to $\mathcal{T}$ when the start time is clear from the context. 
% We say an edge-centric graph (or forest) $F$ is \eeq-equivalent $G$ if the corresponding edges of all nodes in $F$ is \eeq-equivalent to $G$.
Based on the framework described in \ref{subsec:frk}, answering a window query $[ts,te]$ can be considered as searching the \eeq-equivalent structure built for the anchored start time $ts$, where the neighbors of each node in the forest are derived by a binary search. We set the following two optimization goals for high query efficiency:

\begin{itemize}
    \item Minimizing the search space in the forest of the anchored start time;
    \item Minimizing the difference between forests of adjacent start times.
\end{itemize}

% \stitle{Index by Edge Labelling}. 
\noindent
The first objective reduces the number of visited forest nodes in query processing, while the second objective reduces the cost of the binary search for the neighbors of each visited node. To meet both objectives, we begin by introducing the following concept.

% The first reduces the number of nodes visited in the forest, and the second reduces the cost for binary search of neighbors of each visited node. 
% Towards these two goals, we propose an \textit{edge-centric binary forest} structure for the end-anchored $k$-core search.
% \stitle{Warning-Up.} 

\begin{definition}[Edge Core Time]
\label{def:ct}
Given a temporal graph $G$, a fixed start time $ts$, and an integer $k$, the core time of an edge $e (u, v, t)$, denoted as $\ct(e)_{ts}^k$, is defined as the earliest end time $te$ such that the edge $e$ belongs to the $k$-core in the time window $[ts, te]$. 
\end{definition}

We omit the superscript $k$ when it is clear from context. For a fixed start time $ts$, given an arbitrary end time $te$, all edges with core times not later than $te$ are in temporal $k$-cores of $[ts,te]$.

\begin{example}
Refer to \reffig{graph}, for $k = 2$ and the start time $ts = 4$, edge $e_1 = (v_1, v_2, 4)$ has a core time of 4, because the earliest end time at which it appears in the 2-core is $te = 4$. 
Likewise, for the edge $e_2 = (v_6, v_7, 4)$, we have $\ct(e_2)_{4} = 5$. For the time window [4,5], both $e_1$ and $e_2$ have core times no later than 5, therefore they are all in the temporal 2-core of $\mathcal{T}_{[4,5]}$. 
\end{example}

\begin{table}[t!]
  \centering
  \footnotesize
  \renewcommand{\arraystretch}{1.0}
  \setlength{\arrayrulewidth}{0.5pt}
  % scale to the width of one column
  \resizebox{\columnwidth}{!}{%
    \begin{tabular}{r|l|r|l}
      \hline
      $(v_3,v_8,2)$ & $\langle1,5\rangle,\;\langle3, \infty \rangle$
        & $(v_6,v_8,5)$ & $\langle1,5\rangle,\;\langle5, \infty \rangle$ \\
      \hline
      $(v_4,v_5,3)$ & $\langle1,6\rangle,\;\langle4,\infty \rangle$
        & $(v_7,v_8,5)$ & $\langle1,5\rangle,\;\langle5,\infty \rangle$ \\
      \hline
      $(v_1,v_2,4)$ & $\langle1,4\rangle,\;\langle5,\infty \rangle$
        & $(v_2,v_4,6)$ & $\langle1,6\rangle,\;\langle4,\infty \rangle$ \\
      \hline
      $(v_1,v_3,4)$ & $\langle1,4\rangle,\;\langle5,\infty \rangle$
        & $(v_2,v_5,6)$ & $\langle1,6\rangle,\;\langle4,7\rangle,\;\langle5,\infty \rangle$ \\
      \hline
      $(v_2,v_3,4)$ & $\langle1,4\rangle,\;\langle5,\infty \rangle$
        & $(v_5,v_6,7)$ & $\langle1,7\rangle,\;\langle5,\infty \rangle$ \\
      \hline
      $(v_6,v_7,4)$ & $\langle1,5\rangle,\;\langle5,\infty \rangle$
        &                 &                      \\
      \hline
    \end{tabular}
  }
  \caption{The core times for $k=2$ of all edges of $G$ for all start times stored incrementally.}
  \label{tab:ect}
  \vspace{-0.5em}
\end{table}

\begin{example}
\reftab{ect} presents a compact, incremental encoding of every edge’s 2-core time across all start times. For each edge $e$, we store an ordered sequence of pairs $\la ts, \ct(e)_{ts}\ra$, but only emit a new pair when the core time changes or when the edge exits the $k$-core (in which case we set the core time to $\infty$). For example, consider $e = (v_2, v_5, 6)$. Its core time is 6 for $ts = 1, 2, 3$, so we record the first label $\la 1, 6\ra$. At $ts = 4$ the core time increases to 7, so we append $\la 4, 7 \ra$. Finally, at $ts = 5$, the edge leaves the 2-core, so we record, so we record $\la 5, \infty \ra$. Because core times never decrease as $ts$ grows, both the start-time and core-time sequences are monotonic.
\end{example}

% In the start-anchored framework, the PHC Index efficiently stores core times for each vertex by recording only the start times $ts$ where the core time $\text{ct}(u)_{ts}^k $ changes. For instance, for a vertex $u$, the index might be represented as $\text{PHC}[u] = [\langle 1,4 \rangle, \langle 3,5 \rangle]$, where each pair $\langle ts, \text{ct} \rangle$ indicates the core time from $ts$ until the next recorded start time. Here, $\text{ct}(u)_{ts}^k = 4$ for $ts = \{1, 2\}$, and $\text{ct}(u)_{ts}^k = 5$ for $ ts = \{3, 4, 5\}$, assuming a maximum timestamp of 5. From this, we extract the unique core times for $u$: 4 and 5.

Based on \refdef{ct}, a query can be answered by starting a traversal from the query vertex and following only those edges whose core times are no later than the query end time. Core times of all edges take $O(m)$ space. 

\stitle{Reducing Edges from $O(m)$ to $O(n)$.}
For improvement, we observe that not all edges are necessary for query processing.

\begin{definition}[CT-MSF]
\label{def:ctmsf}
Given a temporal graph $G$, an integer $k$, and a start time $ts$, the core-time-based minimum spanning forest (CT-MSF) is a minimum spanning forest of $G$ where the weight of each edge $e$ is its core time for $ts$, i.e., $\ct(e)_{ts}$.
\end{definition}

The CT-MSF is clearly an \eeq-equivalent structure for the fixed start time $ts$, and the following lemma holds.

\begin{lemma}
\label{lem:msf}
Given a CT-MSF $\msf$ of a fixed start time $ts$, for any end time $te$ ($ts \le te$), all edges with core times not later than $te$ in $\msf$ (i.e., $\{e \in \msf | \ct(e)_{ts} \le te\}$) form a spanning forest for the temporal $k$-cores $\mathcal{T}_{[ts,te]}$.
\end{lemma}

% \begin{figure}
%     \centering
%     \includegraphics[scale=0.35]{figs/new/ctmsf.eps}
%     \caption{The CT-MSF of $G$ for the start time $ts = 3$. The label on each edge indicates its core time for $ts = 3$. }
%     \label{fig:ctmsf}
% \end{figure}

\begin{figure}[t!]
  \centering

  \begin{subfigure}[b]{0.3\textwidth}
    \centering
    \includegraphics[width=\textwidth]{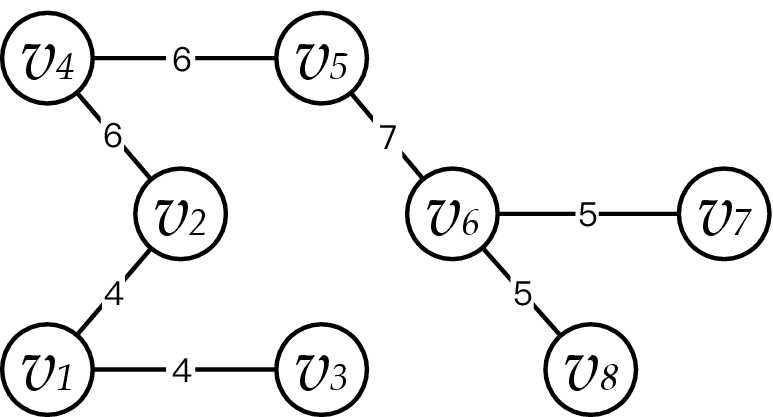}
    \caption{$\msf_3$}
    \label{fig:ctmsf}
  \end{subfigure}%
  \hfill
  \begin{subfigure}[b]{0.45\textwidth}
    \centering
    \includegraphics[width=\textwidth]{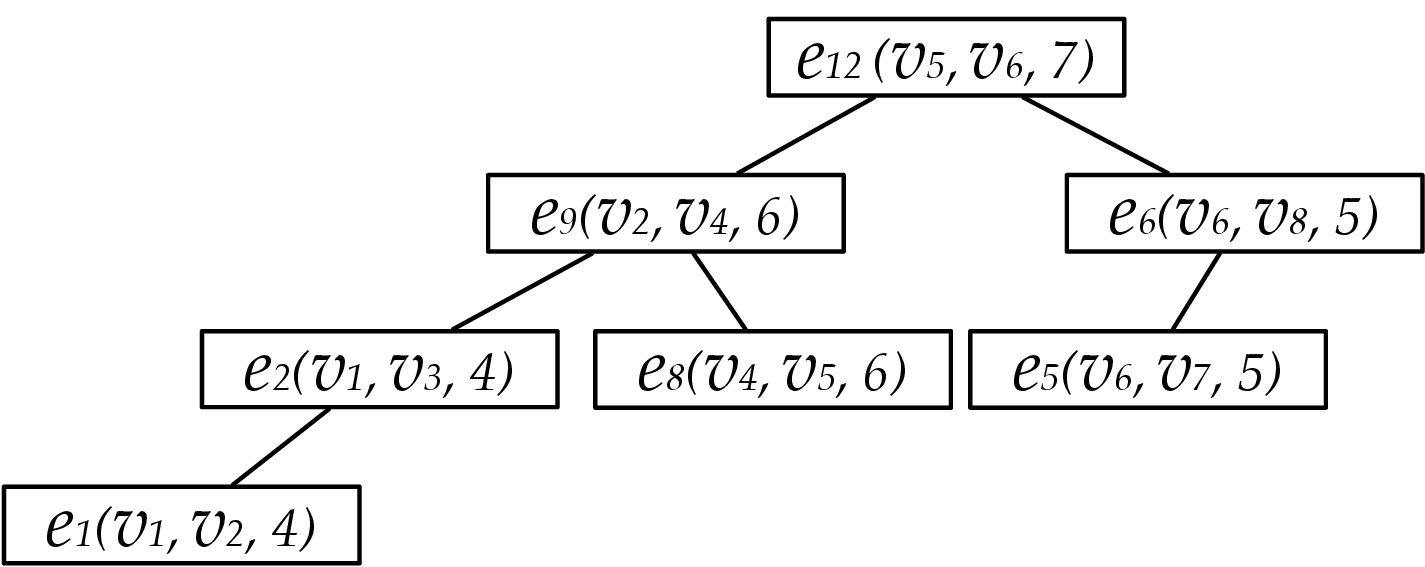}
    \caption{$\binforest_3$}
    \label{fig:ecb}
  \end{subfigure}%

  \caption{The CT-MSF and the corresponding ECB-Forest of $G$ for the start time $ts = 3$. The label on each edge indicates its core time for $ts = 3$.}
\end{figure}

\begin{example}
\reffig{ctmsf} shows the CT-MSF of $G$ for the start time $ts = 3$. It is EC-equivalent to $\mathcal{T}_{[3,t_{max}]}$; every vertex pair that is connected in $\mathcal{T}_{[3,t_{max}]}$ remains connected in $\mathcal{M}_3$, and every pair that is disconnected there is likewise disconnected in $\mathcal{M}_3$.
\end{example}

Converting the core-time-labelled graph into the CT-MSF gives us an $O(n)$ index (edges in a spanning tree are bounded by the number of vertices) for the start-anchored $k$-core search. Therefore, we can run the same searching algorithm on the CT-MSF to get the result.
Yet the CT-MSF still leaves room for optimisation because the number of neighbours per vertex remains uncontrolled.
% However, recall the aforementioned index optimization target, improvement space still exists in the CT-MSF because of the uncontrolled volume of neighbors.
%
First, the query algorithm needs to examine many incident edges at a vertex to locate those whose core times satisfy the window, wasting work on edges and vertices that will never be visited. Second, a vertex needs to record its entire neighbor list in the MSF whenever any single neighbor changes by changing the start time, which is expensive for high-degree vertices.

\stitle{From MSF to Binary Forest.}
Building on the CT-MSF, we propose our edge-centric forest to bound the number of children of each forest node by two (i.e., binary forest).
% The key to improving the performance in the framework is to minimize the structural difference between adjacent end times, which will reduce the index size and improve the efficiency of binary search in query processing. 
% Motivated by this, we design an edge-centric forest structure for the end-anchored connected component search to bound the children number by two (binary forest) for each forest node.
%
By edge-centric, we mean that each forest node represents an edge in the temporal graph, and we adopt edges in the CT-MSF as forest nodes. To avoid ambiguity, we hereafter use the term \textit{vertex} to refer to a vertex in the original graph, and the term \textit{node} to refer to a node in the forest structure. 
The forest relies on a total order of all nodes (i.e., edges in $G$) by core times. Given the start time $ts$, two nodes $eu$ and $ev$ with core times $\ct(eu)_{ts}$ and $\ct(ev)_{ts}$, respectively, we say $eu$ ranks higher than $ev$ if $\ct(eu)_{ts} > \ct(ev)_{ts}$, and the tie is broken by the edge ID. 
% We denote this ranking by $\order(eu) < \order(ev)$ if $eu$ ranks higher than $ev$.
% We impose a total order on all forest nodes (i.e., edges in $G$）at a given start time $ts$: edges are ordered by their core times, and ties are resolved by edge ID.

% \stitle{Aggregating all Start Times.} Note that it is possible for an edge to have multiple core times w.r.t. different start times. Moreover, the core time of an edge never decreases when increasing the start time. Let $\mathcal{C}(e) = \bigl\{\ct(e)_{ts} \mid ts \in [1,t_{\max}]\bigr\}$ denote the set of \emph{distinct} core-time values that $e$ attains. For every unique value $ct \in \mathcal{C}(e)$ we define its \emph{activation time} $\act(e,ct) = \max\bigl\{ts \;\big|\; \ct(e)_{ts} = ct \bigr\}$, that is, the latest start time at which $e$’s core time equals $ct$. 

% We treat every triple $\bigl(e,ct,\act(e,ct)\bigr)$ as a distinct forest node. Consequently, when an edge has more than one core-time value, we create a separate forest node for each value and place them in the same global order we specified earlier.

\begin{definition}[ECB-Forest]
\label{def:forest}
Given a start time $ts$, let $\msf_{ts}$ be a CT-MSF for $ts$. The ECB-Forest (\underline{e}dge-centric \underline{c}ore-equivalent \underline{b}inary forest) for $ts$, denoted by $\binforest_{ts}$, is a binary forest.
Each node $x$ in $\binforest_{ts}$ corresponds to an edge in $\msf_{ts}$. Let $u$ and $v$ be the terminals of the corresponding edge of $x$.
\begin{itemize}
    % \item Left child $l$ of $x$: (1) $l$ is in the same connected component of $u$ in the graph of $\{e \in M_{te} | \order(e) < \order(x) \}$; (2) $\order(l)$ is maximized.

    % \item Right child $r$ of $x$: (1) $r$ is in the same connected component of $v$ in the graph of $\{e \in M_{te} | \order(e) < \order(x) \}$; (2) $\order(r)$ is maximized.

    % \item Left child $l$ of $x$: The \textbf{highest-ranked} node on the $u$-side of $x$ in $\msf_{ts}$ such that: 1) $\order(l) < \order(x)$, and 2) $l$ is connected to $u$ in $\msf_{ts}$ only via edges with lower order than $x$.

    % \item Right child $r$ of $x$: The \textbf{highest-ranked} node on the $v$-side of $x$ in $\msf_{ts}$ such that: 1) $\order(r) < \order(x)$, and 2) $r$ is connected to $v$ in $\msf_{ts}$ only via edges with lower order than $x$.

    % \item Left child $l$ of $x$: (1) $l$ is connected to $u$ in the graph of $\{e \in \msf_{ts} | \order(e) > \order(x) \}$; (2) $\order(l)$ is minimized.

    % \item Right child $r$ of $x$: (1) $r$ is connected to $v$ in the graph of $\{e \in \msf_{ts} | \order(e) > \order(x) \}$; (2) $\order(r)$ is minimized.

    \item Left child $l$ of $x$: (1) $l$ is connected to $u$ in the graph of $\{e \in \msf_{ts} | \ct(e)_{ts} < \ct(x)_{ts} \}$; (2) $\ct(l)_{ts}$ is maximized.

    \item Right child $r$ of $x$: (1) $r$ is connected to $v$ in the graph of $\{e \in \msf_{ts} | \ct(e)_{ts} < \ct(x)_{ts} \}$; (2) $\ct(r)_{ts}$ is maximized.
\end{itemize}
\end{definition}

\begin{example}
\reffig{ecb} shows the ECB-Forest $\binforest_3$ for $G$ at $ts = 3$. In this forest, each node represents an edge of the CT-MSF $\msf_3$, and the parent-child relationships enforce the global edge ranking: every parent node corresponds to a strictly higher-ranked edge, while every child corresponds to a strictly lower-ranked edge.
\end{example}

ECB-Forest is clearly \eeq-equivalent to the temporal $k$-core of the original graph and also guarantees the following elegant properties for query processing.

% \begin{definition}[Edge-Connectivity]
% Given a graph $G$ and a start time $ts$, we say that an edge-centric graph $G'$ of $\mathcal{T}$ is edge-connected if, for any end time $te$, the following holds:
% for every pair of edges in $G$ that belong to the same connected component of $\mathcal{T}_{[ts, te]}$, their corresponding nodes in $G'_{[ts, te]}$ are also connected.
% \end{definition}

\begin{lemma}
\label{lem:one2all}
Given a temporal graph $G$ and a start time $ts$, for any end time $te$, all nodes belonging to the same temporal $k$-core component of $[ts,te]$ are connected in the ECB-forest $\binforest_{ts}$.
\end{lemma}

\reflem{one2all} implies we can detect all nodes (graph edges) in a connected component $C$ by searching from any one node in $C$ in $\binforest_{ts}$. We now discuss the in-memory data structure for the ECB-Forest and the overall index for any query windows followed by our query processing algorithm for the historical $k$-core component search. 

An edge in the graph may have multiple core times w.r.t. different start times. 
We simply regard the same edge with different core times as different parallel edges. In this way, every edge $e$ now has a unique core time, and we use $lst(e)$ to denote the latest start time holding for the core time. Updating the core time of an edge when changing the start time is equivalent to inserting a new edge.

We record the mapping from the ID of each graph edge to a forest node and the reverse mapping. In addition to the core time, for each node $eu$, we use an array to record the neighbors of $eu$ in ECB-Forests for different start times. Specifically, each item in the array of $eu$ has the form of $\la startTime, leftChild, rightChild, parent \ra$; it records, for that start time, $eu$'s parent and its two children in the ECB-Forest. The items in the array is arranged in decreasing order of start times. Note that an item is stored in the array only if the neighbors (a parent and two children) are not the same as the previous start time. We call the overall index for temporal $k$-core component search the \idx (pruned ECB-forest index).

\begin{table}[tbp]
  \centering
  \footnotesize
  \renewcommand{\arraystretch}{1.0}
  \setlength{\arrayrulewidth}{0.5pt}
  \resizebox{\columnwidth}{!}{%
    \begin{tabular}{|c|c|c|l|}
      \hline
      Node & $\ct$ & $lst$ & PECB Index Entries \\
      \hline
      $e_1(v_1, v_2, 4)$ & 4 & 4 &  $\la 4,-,-,e_2 \ra$ \\
      \hline
      $e_2(v_1, v_3, 4)$ & 4 & 4 &  $\la 4,e_1,-,e_{11} \ra$, $\la 3,e_1,-,e_9 \ra$, $\la 2,e_1,-,e_4 \ra$\\
      \hline
      $e_3(v_2, v_3, 4)$ & 4 & 4 &  \\
      \hline
      $e_4(v_3, v_8, 2)$ & 5 & 2 &  $\la 2,e_2,-,e_6 \ra$\\
      \hline
      $e_5(v_6, v_7, 4)$ & 5 & 4 &  $\la 4,-,-,e_6 \ra$\\
      \hline
      $e_6(v_6, v_8, 5)$ & 5 & 4 &  $\la 4,e_5,-,e_{12} \ra$, $\la 2,e_5,e_4,e_9 \ra$\\
      \hline
      $e_7(v_7, v_8, 5)$ & 5 & 4 &  \\
      \hline
      $e_8(v_4, v_5, 3)$ & 6 & 3 &  $\la 3,-,-,e_9 \ra$\\
      \hline
      $e_9(v_2, v_4, 6)$ & 6 & 3 &  $\la 3,e_2,e_8,e_{12} \ra$, $\la 2,e_6,e_8,- \ra$\\
      \hline
      $e_{10}(v_2, v_5, 6)$ & 6 & 3 &  \\
      \hline
      $e_{11}(v_2, v_5, 6)$ & 7 & 4 & $\la 4,e_2,-,e_{12} \ra$ \\
      \hline
      $e_{12}(v_5, v_6, 7)$ & 7 & 4 & $\la 4,e_{11},e_6,- \ra$, $\la 3,e_9,e_6,- \ra$ \\
      \hline
    \end{tabular}
  }
  \caption{The core time, $lst$ (latest start time for the corresponding core time), and PECB‐index labels for each node. }
  \label{tab:index}
\end{table}

\begin{example}

\reftab{index} lists every forest node in the PECB index, showing its core time, its last-start time (for readability only; it is not stored), and the node’s entry array, ordered from lowest to highest in the global ranking. 
For example, edge $e_1(v_1, v_2, 4)$ has core time 4 and a $lst$ time 4. Its PECB entry array contains a single record $\la 4, -, -, e_2\ra$, indicating that at $ts = 4$, $e_1$ has no left child, no right child, and $e_2$ as its parent. 
Edge $e_3$ never appears in any CT-MSF, so its entry array is empty. 
Finally, edge $(v_2, v_5, 6)$ attains two different core-time values at different start times; we treat these as separate forest nodes $e_{10}$ and $e_{11}$. 
\end{example}

\begin{theorem}
The space complexity of \idx is $O(n \cdot \overline{t})$ where $\overline{t}$ denotes the average number of labels for each node.
\end{theorem}

\begin{proof}
The \idx consists of multiple ECB-Forests $\binforest_{ts}$, each corresponding to a distinct start time $ts$. Each $\binforest_{ts}$ is a forest with at most $n - 1$ nodes since it represents a cycle-free MSF. Labels are stored only for nodes in $\binforest_{ts}$ at each $ts$, so for each start time, at most $n - 1$ nodes have labels. Since $\overline{t}$ is the average number of labels per node, the total number of labels is $n \cdot \overline{t}$.
\end{proof}

While $\overline{t}$ is theoretically bounded by the total number of distinct timestamps in the temporal graph, it is typically much smaller in practice. This disparity arises because structural changes in the ECB-Forest only occur when the connectivity is altered by the introduction of new edges. In practice, such changes are relatively infrequent compared to the total number of timestamps.

\begin{algorithm}[t]
\caption{\algquery}
\label{alg:query}
\KwIn{$u, ts, te$ and the \idx}
\KwOut{The connected componet in $G_{[ts,te]}$ containing $u$}
$R \gets \emptyset$\;
$Q \gets$ an empty queue\;
$e \gets$ the lowest ranked node in $\binforest_{ts}$ containing $u$ \;
$Q.push(e)$\;

\While{$Q \neq \emptyset$}{
    $e(v,v',t) \gets Q.pop()$\;
    % \lIf{$e = \kwnull$}{\kwcontinue}
    \lIf{$e$ is visited}{\kwcontinue}
    mark $e$ as visited\;
    % \lIf{$e \not\in \binforest_{te}$}{\kwcontinue}

    $R \gets R \cup \{v,v'\}$\;
    \ForEach{neighbor $e'$ of $e$ in $\binforest_{ts}$ by binary search}{
        \lIf{$\ct(e') \le te$}{$Q.push(e')$}
    }

}

\Return{$R$}
\end{algorithm}

\subsection{The Query Algorithm}

We present our query processing algorithm in \refalg{query}. 
% This algorithm is designed to efficiently retrieve the temporal $k$-core component of a given query vertex $u$ and integer $k$ within a specified time window $[ts, te]$ in a temporal graph.
The idea is to identify the relevant edges (represented as nodes in an ECB-Forest) connected to $u$, perform a Breadth-First Search (BFS) traversal across the ECB-Forest to explore all reachable nodes, and  return the vertices associated with these nodes as the result.

The algorithm begins by identifying the lowest ranked forest nodes connected to the query vertex $u$ within $\binforest_{ts}$ (lines 3). By maintaining a direct lookup from each start time to its corresponding entry nodes, this step can be performed in logarithmic time. Next, we initialize a BFS traversal by enqueuing these forest nodes into a queue $Q$ (lines 5-11). While the queue $Q$ is non-empty, we dequeue a node $e$ and examine its neighbors in the ECB-Forest for the specified start time $ts$. These neighbors may include parent or child nodes in the forest structure. For each neighbor $e'$, we perform a binary search to retrieve its corresponding label to the query time window $[ts, te]$ (line 10). The corresponding label is the label with the smallest start time $ts'$ such that $ts' >= ts$. If a valid label is found for a neighbor $e'$, we enqueue it into $Q$ for further exploration. This process repeats, allowing the BFS to systematically traverse all reachable forest nodes in the ECB-Forest.

During the BFS process, the algorithm extracts all vertices from the edges represented by the forest nodes. These vertices collectively form the connected component of the query vertex $u$ within the temporal $k$-core $\tcore$, which is returned as the final result.

\begin{example}
Given $k = 2$, time window $[3,5]$, and the query vertex $v_2$, we first retrieve the lowest ranked node in $\binforest_3$, which is $e_2$. The corresponding index label $\la 3, e_1, -, e_9 \ra$ indicates that $e_2$ is adjacent to $e_1$ and $e_9$. We enqueue these two nodes.
Next, we dequeue $e_1$ and read its label $\la 4, -, -, e_2 \ra$, which adds $v_1$ and $v_3$ to the result set but reveals no new neighbours. Finally, we dequeue $e_9$, but see that its core time of 6 exceeds the query’s end time $te = 5$, so we discard it without further exploration. 
The queue is now empty, and the vertices discovered, namely $v_1$, $v_2$ and $v_3$, form the desired 2-core component in $\mathcal{T}_{[3,5]}$.
\end{example}

\begin{theorem}
\label{lem:query_time}
The time complexity of \refalg{query} is $O(r \cdot \log{\overline{t}})$, where $d$ is the average degree and $r$ is the result size.
\end{theorem}
\vspace{-0.5em}

\begin{proof}
Locating the node connected to the query vertex $u$ with the lowest ranked takes $O(\log \overline{t})$ time, where $\overline{t}$ is the average number of labels. The BFS traversal in the ECB-Forest visits each node in the connected component once, with the number of nodes visited being proportional to $r$. For each node, a binary search is performed on its labels to validate the query time window, costing $O(\log \overline{t})$ per node. Thus, the total time complexity is $O(r \cdot \log \overline{t})$.
\end{proof}

%!TEX root = main.tex
\section{Index Construction}
\label{sec:construct}

\stitle{Computing Edge Core Times.} Given an integer $k$, core times of all vertices for all possible start times can be efficiently computed by the algorithm in \cite{DBLP:journals/pvldb/YuWQ00021}. We first briefly review the algorithm for self-completeness. Then we discuss how to extend it to derive the edge core times. 
To derive vertex core times, the algorithm first computes the core times of vertices for the earliest start time by iteratively removing edges with the latest timestamps. The core time of a vertex is $t$ if it leaves the $k$-core after removing the edge at $t$. Then, starting from the earliest start time, the algorithm updates the core time of vertices by iteratively increasing the start time. Certain values for each vertex are maintained in constant time when increasing the start time and help track if the core time needs update. If so, the algorithm recomputes the core time of the vertex for the current start time, which needs to scan the core times of all neighbors in linear time. Many vertices may not update core times when increasing the start time, and the algorithm only outputs core times of each vertex if the value is different from that of the previous start time in the iteration. The time complexity of the algorithm is bounded by $O(vct \cdot d)$, where $vct$ and $d$ represent the size of the output and the average degree, respectively.

The core time of any edge for the start time $ts$ is the larger one among the core times of its terminal vertices for $ts$. Therefore, when the core time of a vertex $u$ changes in the algorithm, we simply update the core times of all edges connecting to $u$. In this way, we derive core times of all edges as byproducts when the algorithm terminates. The time complexity does not change since we always need to scan all neighbors of a vertex when its core time changes.

\stitle{Idea for ECB-Forest Maintenance.}
We propose a novel incremental algorithm to construct the index. Given an ECB-Forest for a start time $ts$, we study an algorithm to update the forest given a new edge $e$ changing its core time at $ts - 1$. In this way, we can initialize an empty ECB-Forest for the largest start time and derive the ECB-Forests by inserting edges. For the ECB-Forest of each start time, we only store the difference compared with the previous one as discussed in \refsubsec{index}. 
% It is clear to see that the method is straightforward to handle the continuously arriving edges because a new edge always arrives later than existing edges.

We now discuss how to update the ECB-Forest $\binforest$ given a new edge. 
Below, we use $ts$ to denote the start time before inserting $e$, i.e., $ ts = lst(e) + 1$. Based on \reflem{msf}, simply adding the new edge $e$ into $\binforest$ already yields a forest structure that is \eeq-equivalent to the graph at time $ts-1$. Hence, we first create a forest node for $e$ and insert it into $\binforest$ according to the global edge order. We also use $e$ to denote the corresponding forest node of the new edge. Note that adding $e$ may create a cycle within the current $\binforest$, and hence in the corresponding MSF. In that case, the edge with the highest rank on the cycle must be removed from $\binforest$ after the insertion of $e$. Two outcomes are possible:

\begin{itemize}
  \item  If the expired edge is an existing edge $e'\neq e$, $e'$ is removed from $\binforest$.

  \item If $e$ itself is the highest-ranked edge in the cycle, the deletion step removes the node we have just inserted, and $\binforest$ returns to its previous state.
\end{itemize}

In both cases the expired node will be safely removed during the construction process of adding the new edge. With the new node $e$, the main challenges in maintaining the ECB-Forest are: (1) how to connect the new node $e$ to the existing ECB-Forest to produce an EC-equivalent graph and (2) how the structure updates to a valid ECB-Forest accordingly.

\begin{algorithm}[t]
\caption{\algfind}
\label{alg:find}
\KwIn{$e(u,v)$}
\KwOut{$l$, $r$, $eu$, $ev$}
    % $l \gets$ the highest-ranked edge connecting to $u$ in $\binforest$ s.t. $\order(l) > \order(e)$\;
    % $r \gets$ the highest-ranked edge connecting to $v$ in $\binforest$ s.t. $\order(r) > \order(e)$\;
    % $eu \gets$ the lowest-ranked edge connecting to $u$ in $\binforest$ s.t. $\order(eu) < \order(e)$\;
    % $ev \gets$ the lowest-ranked edge connecting to $v$ in $\binforest$ s.t. $\order(ev) < \order(e)$\;
    
    $l \gets$ the highest-ranked edge connecting to $u$ in $\binforest$ s.t. $\ct(l) < \ct(e)$\;
    $r \gets$ the highest-ranked edge connecting to $v$ in $\binforest$ s.t. $\ct(r) < \ct(e)$\;
    $eu \gets$ the lowest-ranked edge connecting to $u$ in $\binforest$ s.t. $\ct(eu) > \ct(e)$\;
    $ev \gets$ the lowest-ranked edge connecting to $v$ in $\binforest$ s.t. $\ct(ev) > \ct(e)$\;

    \If {$l$ exists} {
        \While {$l$ has a parent and $\ct(l.parent) < \ct(e)$} {
            $l \gets l.parent$\;
        }
        \If {$l.parent < eu$} {
            $eu \gets l.parent$\;
        }
    }
    update $r$ and $ev$ by performing lines 5-9\;
    \textbf{Return} $l, r, eu, ev$\;
\end{algorithm}

\begin{figure*}[tbp]
\centering
\begin{subfigure}[b]{0.23\linewidth}
    \includegraphics[width=\linewidth]{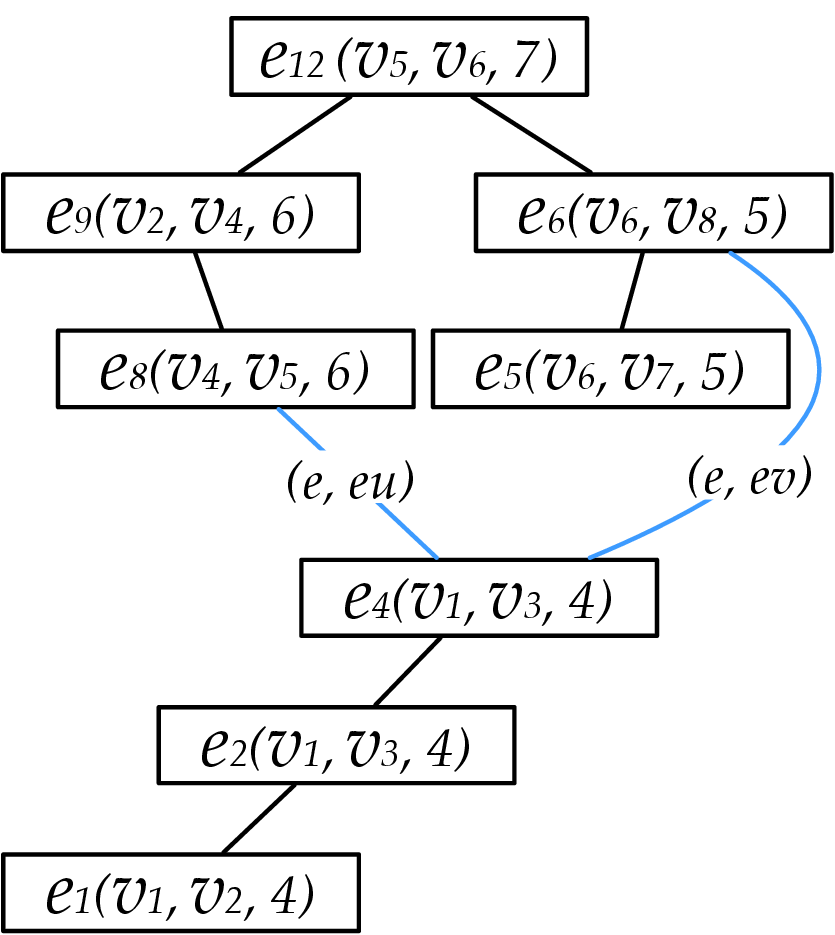}
    \caption{Node insertion: $F_1$.}
    \label{fig:step_1}
\end{subfigure}
\hfill
\begin{subfigure}[b]{0.23\linewidth}
    \includegraphics[width=\linewidth]{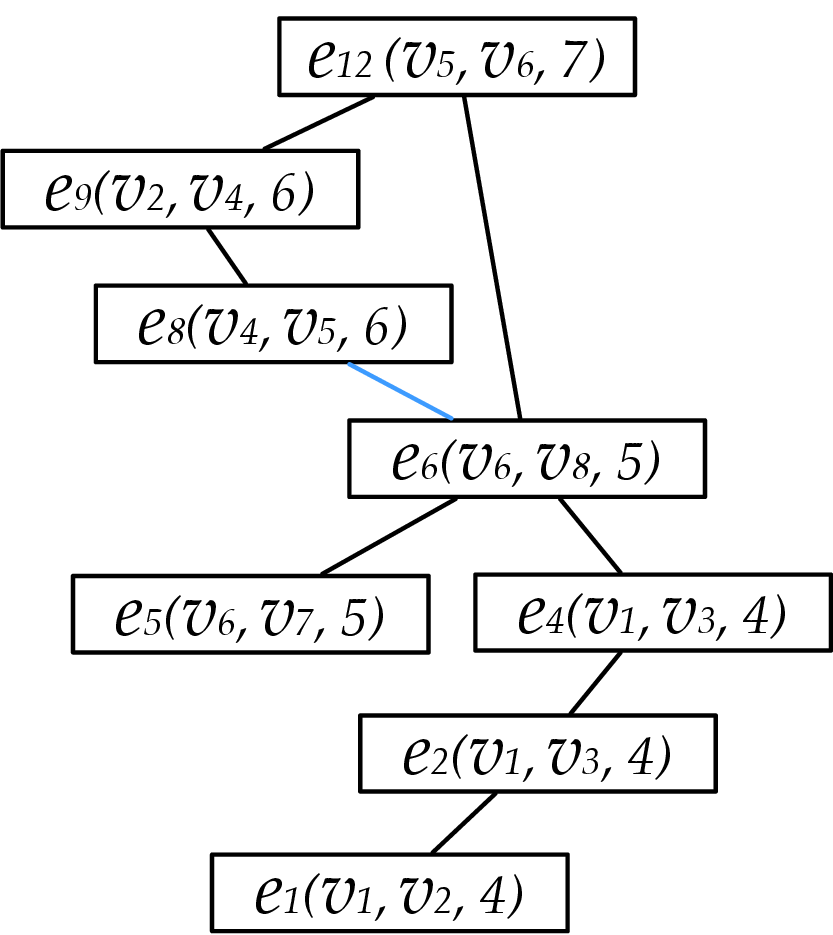}
    \caption{WE-operation: $F_2$.}
    \label{fig:step_2}
\end{subfigure}
\hfill
\begin{subfigure}[b]{0.23\linewidth}
    \includegraphics[width=\linewidth]{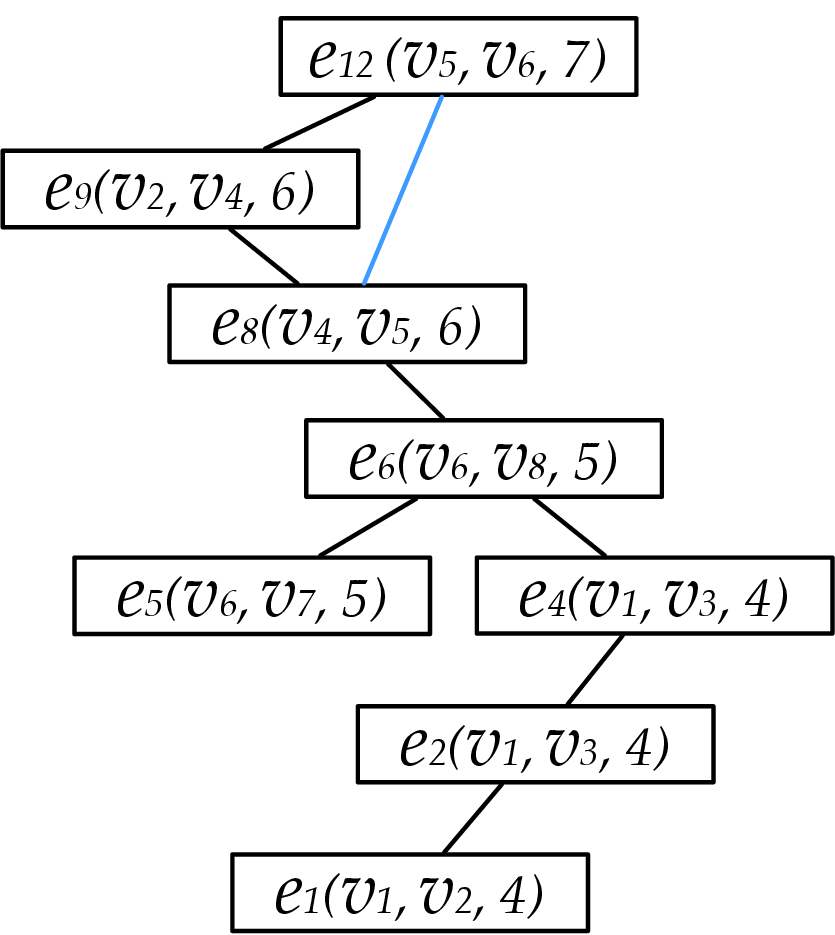}
    \caption{WE-operation: $F_3$.}
    \label{fig:step_3}
\end{subfigure}
\hfill
\begin{subfigure}[b]{0.23\linewidth}
    \includegraphics[width=\linewidth]{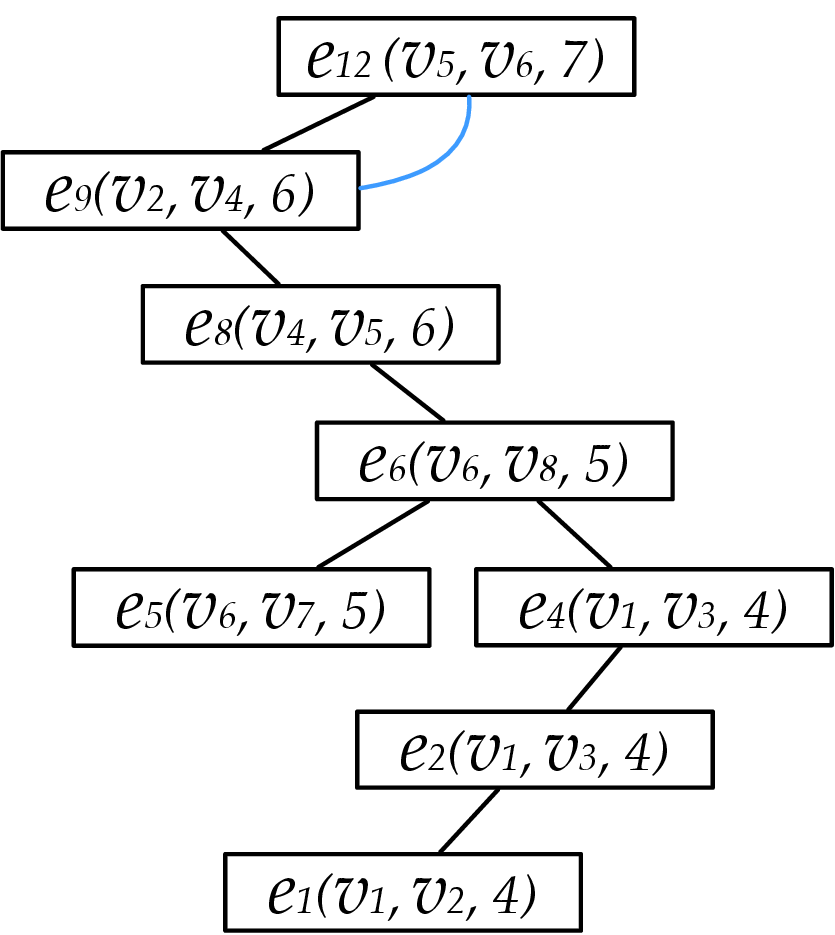}
    \caption{WE-operation: $F_4$.}
    \label{fig:step_4}
\end{subfigure}
\caption{An illustration of the process of adding the edge $e_{4} = (v_1, v_3, 4)$ into the ECB-Forest $\binforest_3$ to obtain $\binforest_2$. New forest edges of each step are colored in blue.}
\label{fig:te4}
\end{figure*}

\stitle{Finding the Correct Position.}
To address the first challenge, we first connect $e$ to certain existing forest nodes. To correctly insert a new edge $e$, it must be positioned relative to existing edges with both higher and lower ranks. We first identify the left child $l$ and right child $r$ of $e$ based on \refdef{forest}, and set the insertion points $eu$ and $ev$ as the parents of $l$ and $r$, respectively. If $l$ does not exist, we instead set $eu$ to be the incident edge on $u$ with the lowest rank among all edges whose rank exceeds that of $e$. The same applies to $ev$ if $r$ does not exist. This ensures that $e$ is anchored to forest nodes with lower rank via its children $l$ and $r$, while its parent links (to $eu$ and $ev$) connect it to the forest structure with higher rank.

For simplicity, we use $\oplus$ and $\ominus$ to denote operators of inserting a tree edge and deleting a tree edge, respectively. For example, $F \oplus (e,e')$ means connecting the nodes $e$ and $e'$ in $F$. We propose the following lemma to guarantee that inserting $e$ in this manner preserves the structural properties of the ECB-Forest:

% \begin{lemma}
% \label{lem:addedge}
% Given an ECB-Forest $\binforest$
% % Given a \eeq-equivalent edge-connected edge-centric graph (or forest) $F$ 
% of a temporal graph $G$ for the end time $te$ and a new edge $e(u,v,t)$ with $t > te$, let $eu$ and $ev$ be the nodes (edges in $G$) in $\binforest$ containing $u$ and $v$ with the smallest orders, respectively. $\binforest \oplus (e,eu) \oplus (e,ev)$ is an \eeq-equivalent and edge-connected graph of $G$ for the end time $t$.
% \end{lemma}

\begin{lemma}
\label{lem:addedge}
Given an ECB-Forest $\binforest$ of a temporal graph $G$ for the start time $ts$, a new edge $e$ with $lst(e) = ts - 1$, and the insertion points $l$, $r$, $eu$, and $ev$ as defined above, the updated forest $\binforest \oplus (e, l) \oplus (e, r) \oplus (e, eu) \oplus (e, ev) \ominus (l, eu) \ominus (r, ev)$ remains EC-equivalent and edge-connected for $ts - 1$.
\end{lemma}

% \begin{proof}
% Since $\binforest$ is edge-connected for $te$, it satisfies this for all $ts \leq te$. For the new end time $t$, $G[ts, t]$ includes all edges in $G$ with timestamps from $ts$ to $t$, potentially increasing connectivity due to $e$. By adding $(e, eu)$ and $(e, ev)$, the node $e$ connects to $eu$ and $ev$, which represent the latest edges in $\binforest$ involving $u$ and $v$. If $e$ connects previously disjoint components in $G[ts, t]$ via the path $u \xrightarrow{e} v$, the forest now links the components of $eu$ and $ev$ through $e$, ensuring that nodes corresponding to edges in the same component remain connected. Moreover, since $\binforest$ is EC-equivalent for $te$, and the addition of $(e, eu)$ and $(e, ev)$ only reflects the new connectivity introduced by $e$ at time $t$, no additional paths are created beyond those in $G[ts, t]$. Thus, the updated structure maintains equivalence to the temporal graph’s connectivity.
% \end{proof}

\begin{proof}
Adding edges $(e, l)$, $(e, r)$, $(e, eu)$, and $(e, ev)$ integrates $e$ into $\binforest$, connecting it to the forest structure. Removing $(l, eu)$ and $(r, ev)$ detaches prior connections, positioning $e$ as a bridge between sub-trees previously linked through $eu$ and $ev$. For any $ts \leq te$, the adjusted connections ensure that all edges in the same component as $e$ remain reachable, preserving the connected components of $G$. For any pair of edges in the same component in $G$’s projection, the updated paths through $e$ ensure their nodes remain connected in $\binforest$, as $e$ preserves the reachability previously provided by $(l, eu)$ and $(r, ev)$.
\end{proof}

\refalg{find} presents the procedure for identifying the insertion points $l$, $r$, $eu$, and $ev$. We first initialize these insertion points by examining the neighbors of $e$ in $\binforest$ (lines 1-4). Specifically, we search for the highest-ranked edge with a rank lower than $e$ to obtain initial candidates for $l$ and $r$. Similarly, we search for the lowest-ranked edge with a rank higher than $e$ as initial candidates for $eu$ and $ev$. After obtaining these initial values, we traverse upward along the parent chain of $l$, updating its parent to the highest-ranked ancestor whose rank remains lower than $e$ (lines 6--7). This iteration terminates once the final value of $l$ is determined. We then set $eu$ as the parent of the finalized $l$ (lines 8--9). A similar procedure is applied to obtain $r$ and $ev$.

\begin{lemma}
The time complexity of \refalg{find} is $O(h)$, where $h$ is the depth of the ECB-Forest.
\end{lemma}

\begin{proof}
The initialization of the insertion points can be performed in constant time. Lines 5–9 may involve traversing upward in the tree, which incurs a time complexity of $O(h)$.
Thus, the overall time complexity is $O(h)$.
\end{proof}

\begin{example}
At start time $ts = 2$, \reffig{step_1} illustrates the process of inserting a new edge $e_{4} = (v_1, v_3, 4)$ (with $lst = 2$) into the ECB-Forest $\binforest_3$. Following our insertion rule, we identify $l = e_2$, $eu = e_8$, and $ev = e_6$. We first connect $e_2$ to $e_4$ as its new parent, and disconnect $e_2$ from its previous parent $e_9$. We then connect $e_{4}$ to both $e_8$ and $e_6$. The resulting forest, denoted as $F_1 = \binforest_3 \oplus (e_2, e_4) \oplus (e_{10}, e_8) \oplus (e_{10}, e_6) \ominus (e_2, e_9)$, is \eeq-equivalent and edge-connected with respect to $G$.
\end{example}

% As shown in the above example, applying the \reflem{addedge} to the existing ECB-Forest will produce a cycle. 

\stitle{Cycle Elimination.} As discussed earlier, adding a new node via \reflem{addedge} derive an EC-equivalent graph but produce a cycle.
We next introduce a set of operators that break the newly formed cycles and restore the structure to a binary forest after a new node is inserted, thus resolving the second challenge outlined above.

\begin{definition}[WE Operator]
\label{def:we_operator}
Given an edge-centric graph $F$ of $G$, and an arbitrary wedge (two-hop path) $\la (ex,ez), (ez,ey) \ra$ with $\ct(ex) > \ct(ey) > \ct(ez)$, the WE (\underline{W}edge \underline{E}quivalence) operator updates the wedge to $\la (ex,ey), (ey,ez) \ra$, i.e., $F \oplus (ex,ey) \ominus (ex,ez)$.
\end{definition}

\begin{lemma}
\label{lem:we}
Given an edge-centric \eeq-equivalent edge-connected graph $F$ of $G$, $F$ is still \eeq-equivalent and edge-connected after executing the WE Operator on an arbitrary wedge in $F$.
\end{lemma}

% \begin{lemma}
% Assume there exists only one cycle $C$ in an edge-centric graph $F$. After executing a WE operator on a wedge $\langle (ex,ey), (ey,ez) \rangle$ with $\timestamp(ex) > \timestamp(ez)$, there exists 
% \end{lemma}
% a proof to show executing we operator would not produce a cycle.

% Assume there exists only one cycle $C$ in an edge-centric graph $F$. 
We say a wedge of a node $e$ if $e$ is the common node in the two edges of the wedge. Based on \reflem{we}, our idea is to continuously identify the node with the lowest rank in the cycle and apply the WE operator on the wedge of the node. When a new edge between two nodes is created in the WE operator, we set the node with the higher rank as the parent of the other node. Each execution of WE operator will reduce the cycle length by one, and no new cycle will be produced.

\begin{example}
Refer to \reffig{step_1}, where the newly arrived edge $e_4$ is inserted into $\binforest_3$. This insertion creates a new cycle consisting $\{e_4, e_8, e_9, e_{12}, e_6\}$. We identify the wedge as $(e_4, e_6)$ and $(e_4, e_8)$, and apply a wedge elimination (WE) operation to transform it into $\la (e_4, e_6), (e_6, e_8) \ra$. The result is shown in \reffig{step_2}, where the cycle is shortened to $\{e_6, e_8, e_9, e_{12}\}$.
\end{example}

Before the last execution of the WE operator, the cycle has been reduced to a triangle. Executing the WE operator on the wedge of the lowest-rank node in the cycle will produce a pair of parallel edges. We can simply keep one edge, and the cycle is eliminated now.

A byproduct of the process of cycle elimination is to identify the expired node in the ECB-Forest. Recall that the insertion of the new forest node (MSF edge) will exclude an old forest node from the ECB-Forest. We identify the expired node as follows.

\begin{lemma}
\label{lem:old}
Given $e$, $eu$, $ev$ defined in \reflem{addedge}, the lowest common ancestor (LCA) of $eu$ and $ev$ is not in the MSF for the new start time $ts - 1$.
\end{lemma}

Let $elca$ be the LCA of $eu$ and $ev$ mentioned in \reflem{old}. It is easy to see that $elca$ is the node with the highest rank in the cycle.
In each execution of the WE operator, the node with the highest rank is excluded from the cycle. Consequently, we can locate $elca$ as the last visited node when eliminating the cycle (i.e., the terminal node of the parallel edges with the earlier timestamp).
To delete $elca$ from the updated ECB-Forest while bound the number of children of each forest node, we can simply connect the child of $elca$ to its parent because $elca$ only has one child after the execution of all WE operators. 

\begin{example}
In \reffig{step_4}, we obtained the edge-centric graph $F_4$ after applying three WE-operations.
At this point, a cycle consisting of two nodes, $\{e_9, e_{12}\}$, remains, which also corresponds to a pair of parallel edges. We identify $e_{12}$ as the lowest common ancestor (LCA) of $e_8$ and $e_6$, and remove $e_{12}$ to restore the ECB-Forest for $ts = 2$.
\end{example}

\begin{algorithm}[t]
\caption{\algconstruct}
\label{alg:core_construct}
\KwIn{$G$, $k$}
\KwOut{The PECB-Index for $k$-core CC}
    compute the edge core times for $e \in E$\;

    \For {$t_{max} \geq ts \geq 1$} {
        \ForEach{$x(u,v)$: $\ct(x)_{ts} \neq \ct(x)_{ts+1}$} {
            $l, r, eu, ev \gets \algfind(x)$\;
            $l.parent \gets x$\;
            $r.parent \gets x$\;
            Update $\binforest$ based on lines 6-7\;
            $\algmerge(x,eu,ev)$\;
        }
    }
    
    \proc{$\algmerge(e,eu,ev)$}{
    \If{$eu = ev$}{
        \If{$eu \neq \kwnull$}{
            $e.parent \gets eu.parent$\;
            delete $eu$\;
        }
        \Return
    }
    \eIf{$eu = \kwnull \lor \ct(eu) > \ct(ev)$}{
        $e.parent \gets ev$\;
        $\algmerge(ev,ev.parent,eu)$\;

    }{
        $e.parent \gets eu$\;
        $\algmerge(eu,eu.parent,ev)$\;
    }
}
\end{algorithm}

\stitle{The Final Algorithm.}
\refalg{core_construct} constructs the PECB-Index for a temporal graph. It begins with an empty ECB-Forest and processes all edges in decreasing order of start times, incrementally building and updating the forest for each new start time $ts$. The process starts at $ts = t_{max}$ and continues until all edges in $G$ have been processed. For each edge $e$, the algorithm updates the ECB-Forest corresponding to $ts = lst(e)$, ensuring that the index reflects the temporal graph's connectivity up to that point. 

We first compute the edge core times (line 1). Then, we decrease the start time from $t_{max}$ to $1$. For each start time $ts$, we process all edges with $lst = 2$. The insertion points for each edge (forest node) $e$ are determined using \refalg{find} (line 4). We first connect $e$ to $l$ and $r$, establishing its connections to lower-ranked edges in the forest structure (lines 5–6). We then update the PECB-Index by applying the operations specified in \reflem{addedge} (line 7). Finally, we invoke the $\algmerge$ procedure to connect $e$ to higher-ranked parts of the forest and eliminate potential cycles (line 8). The Merge procedure handles cycle elimination by employing the WE-operation (lines 15-20), which identifies the node with the highest rank within the cycle (the LCA edge) and removes it. The node to be removed is identified as the endpoint of parallel tree edges (lines 10-13); by removing the LCA node, it breaks the cycle while maintaining the forest's connectivity equivalence. Once these structural adjustments are made, the algorithm records labels for the affected nodes in the index. \reflem{addedge} and \reflem{we} already guarantees the resulting structure is \eeq-equivalent and edge-connected. 

\begin{theorem}
The time complexity of \refalg{core_construct} is $O(vct \cdot d + |E_{ct}| \cdot h)$, where $vct \cdot d$ corresponds to the cost of computing the core-time graph, $|E_{ct}|$ is the number of distinct edge core time instances (i.e., the number of forest nodes), and $h$ is the depth of the ECB-Forest.
\end{theorem}

\begin{proof}
Computing the core-time graph takes $O(vct \cdot d)$.  
We then process $|E_{ct}|$ forest nodes, where processing each node requires $O(h)$ to find the insertion points (line 4) and for the merge procedure (line 8).  
Therefore, the overall time complexity is $O(vct \cdot d + |E_{ct}| \cdot h)$.
\end{proof}

\section{Experiments}
\label{sec:experiments}

\begin{table}[t]% h asks to places the floating element [h]ere.
  
  \resizebox{\linewidth}{!}{
  \begin{tabular}{lrrrrr}
    \toprule
    Name & $|V|$ & $|E|$ & $t_{max}$ & $k_{max}$ & $day$\\
    \midrule
    FB-Forum (FB) & 899 & 33,786 & 33,482 & 19 & 164\\
    BitcoinOtc (BO) & 5,881 & 35,592 & 35,444 & 21 & 1903\\
    CollegeMsg (CM) & 1,899 & 59,835 & 58,911 & 20 & 193\\
    Email (EM) & 986 & 332,334 & 207,880 & 34 & 803\\
    Mooc (MC) & 7,143 & 411,749 & 345,600 & 76 & 29\\
    MathOverflow (MO) & 24,818 & 506,550 & 505,784 & 78 & 2350\\
    AskUbuntu (AU) & 159,316 & 964,437 & 960,866 & 48 & 2613\\
    Lkml-reply (LR) & 63,399 & 1,096,440 & 881,701 & 91 & 2921\\
    Enron (ER) & 87,273 & 1,148,072 & 220,364 & 53 & 16217\\
    SuperUser (SU) & 194,085 & 1,443,339 & 1,437,199 & 61 & 2773\\
    WikiTalk (WT) & 1,219,241 & 2,284,546 & 1,956,001 & 68 & 4762\\
    Wikipedia (WK) & 91,340 & 2,435,731 & 4,518 & 117 & 5077\\
    ProsperLoans (PL) & 89,269 & 3,394,979 & 1,259 & 111 & 2142\\
    Youtube (YT) & 3,223,589 & 9,375,374 & 203 & 88 & 225\\
    DBLP (DB) & 1,824,701 & 29,487,744 & 77 & 286 & 29219\\
  \bottomrule
  \end{tabular}
  }
  \caption{Datasets.}
  \label{tab:datasets}
  \vspace{-1em}
\end{table}

We conduct extensive experiments to evaluate the performance of our proposed solutions. All algorithms are implemented in C++ and compiled using the g++ compiler with the -O3 optimization level. Experiments are conducted on a Linux machine equipped with dual Intel Xeon Gold 6342 2.8 GHz CPUs and 512 GB of RAM. We evaluate performance on 15 publicly available real-world temporal graphs from SNAP\footnote{\url{https://snap.stanford.edu/}}, the KONECT\footnote{\url{http://konect.cc/}} project and the Network Repository\footnote{\url{https://networkrepository.com/}}. Detailed statistics of these graphs are summarized in \reftab{datasets}.

\begin{figure*}[t]
\centering
\includegraphics[width=\textwidth]{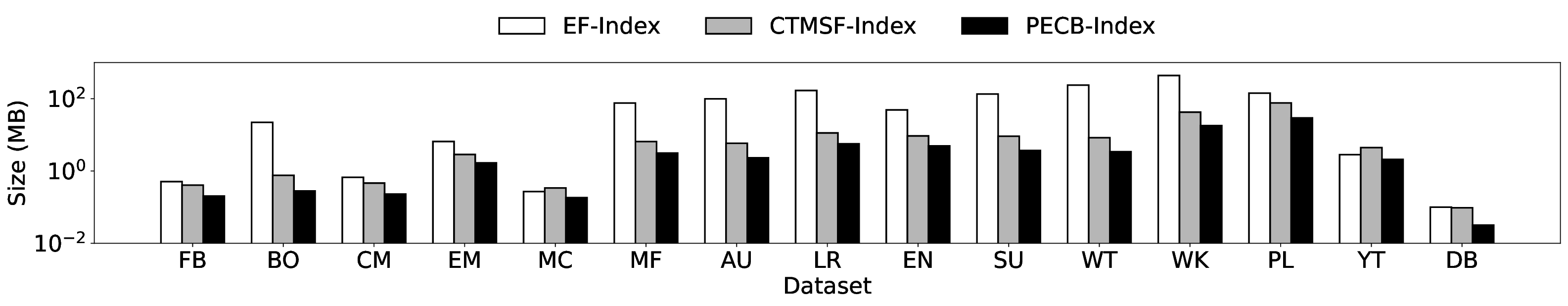}
\vspace{-2em}
\caption{Index space for datasets with timestamps by days.}
\label{fig:space_day}
\vspace{-0.5em}
\end{figure*}

\begin{figure*}[t]
\centering
\includegraphics[width=\textwidth]{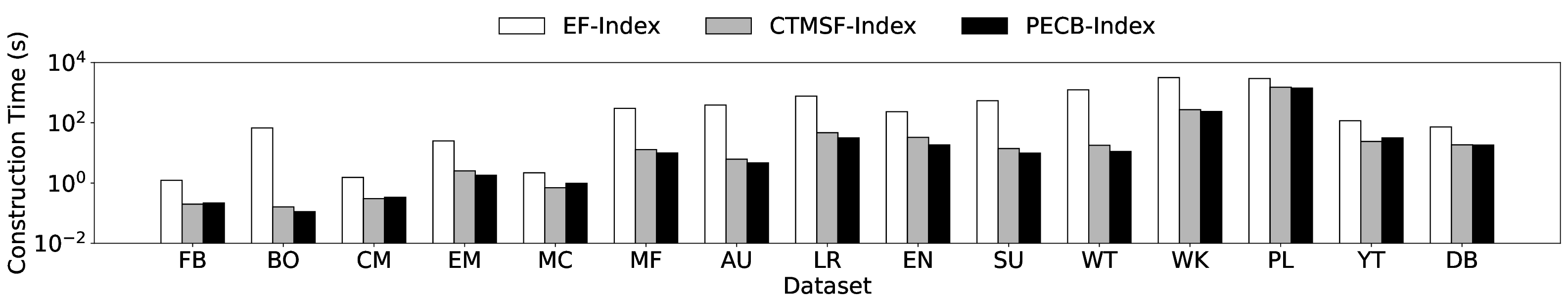}
\vspace{-2em}
\caption{Construction time for datasets with timestamps by days.}
\vspace{-0.5em}
\label{fig:construction_day}

\end{figure*}

\begin{figure*}[t]
\centering
\includegraphics[width=\textwidth]{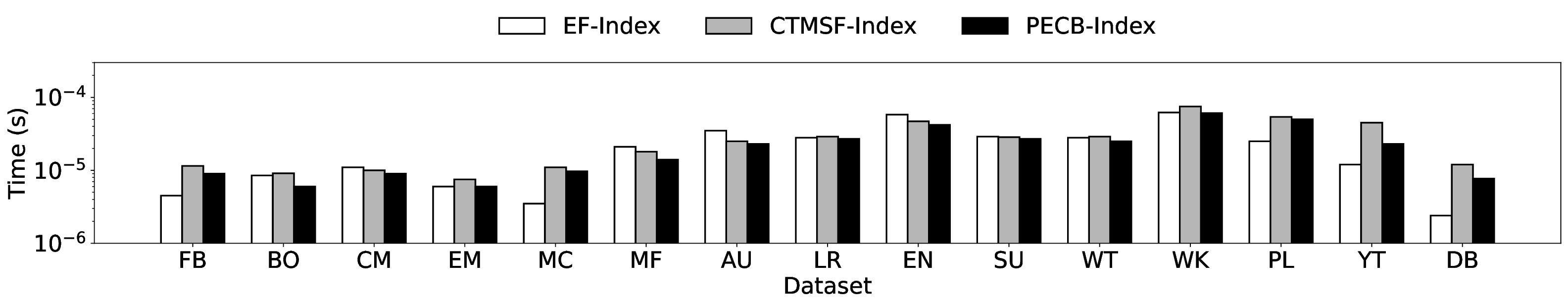}
\vspace{-2em}
\caption{Average query time on datasets with timestamps by days.}
\vspace{-0.5em}
\label{fig:query_day}

\end{figure*}

\stitle{Parameters.} We evaluate the performance of historical $k$-core component search algorithms in terms of query processing time, index construction time, and index size. As in \cite{yang2024}, we first aggregate all edges that fall within the same calendar day and use that day as a single timestamp. Day-level granularity is often more practical, as it reflects the natural reporting cycle of most datasets, and greatly reduces both the number of distinct time points and the size of the index, hence speeding up queries. Nevertheless, to show that our method also works on finer time scales, we also compare the performance os both algorithms on the dataset with original timestamps.  Following \cite{yang2024}, we set the value of $k$ to 50\%, 60\%, 70\%, 80\% and 90\% of $k_{max}$, where $k_{max}$ is the largest valid $k$ value. We set 70\% $k_{max}$ as the default value for $k$. To evaluate query performance, we generate 1,000 queries by randomly selecting vertices and time windows $[ts, te]$. We impose a 24-hour time limit and a 200 GB memory cap; any experiment that exceeds either bound is recorded as a failure to complete.

We compare our proposed PECB-Index with two baselines:
\begin{itemize}
    \item \textbf{EF-Index} \cite{yang2024}: the current state-of-the-art index for historical $k$-core component search. 
    \item \textbf{CTMSF-Index}: our vertex-centric baseline implementation. Instead of applying the edge centric framework, we materialise the CT-MSF (\refdef{ctmsf}) by letting each vertex store the list of incident MSF edges. As in PECB-Index, a vertex writes a new list only when it differs from the previous start time; nonetheless, vertex degree in a CT-MSF is unbounded.
\end{itemize}

% So high-degree vertices retain long, frequently updated neighbour lists, inflating space and slowing access. PECB-Index eliminates this overhead by switching to an edge-centric binary forest that limits every node to two children and records changes only when they occur.

\begin{figure*}[htbp]
  \centering
  \begin{subfigure}[b]{0.24\textwidth}
    \centering
    \includegraphics[width=\textwidth]{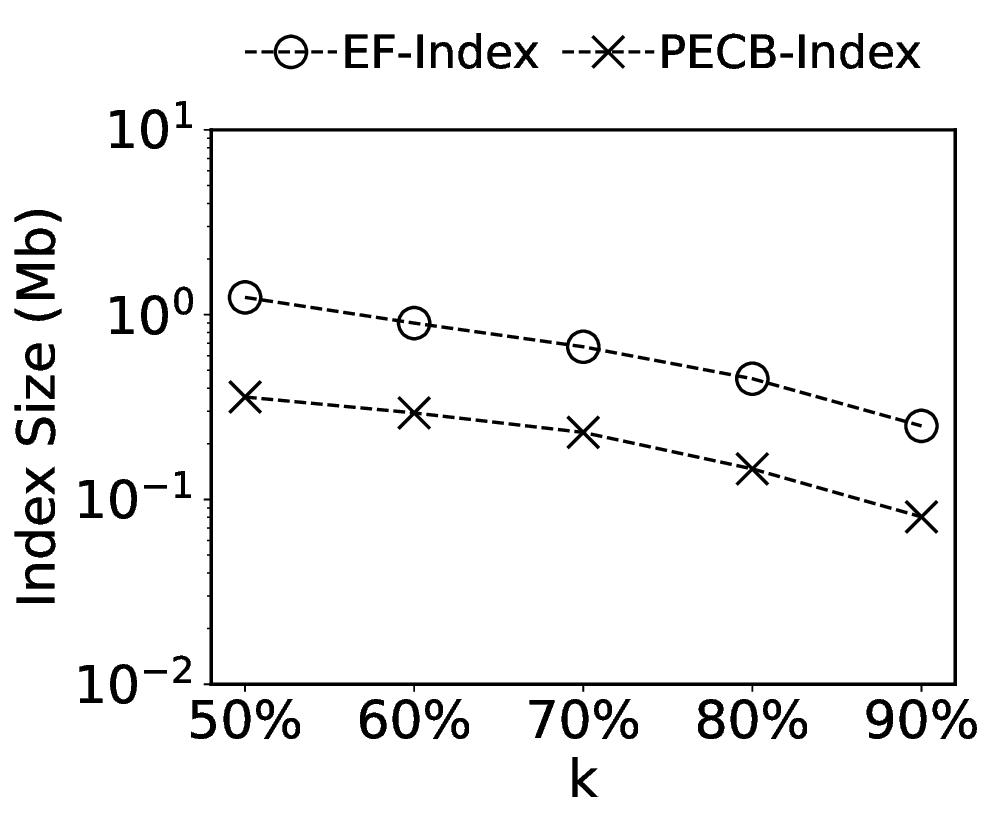}
    \caption{CollegeMsg}
    \label{fig:cm_vary_k_query}
  \end{subfigure}%
  \hfill
  \begin{subfigure}[b]{0.24\textwidth}
    \centering
    \includegraphics[width=\textwidth]{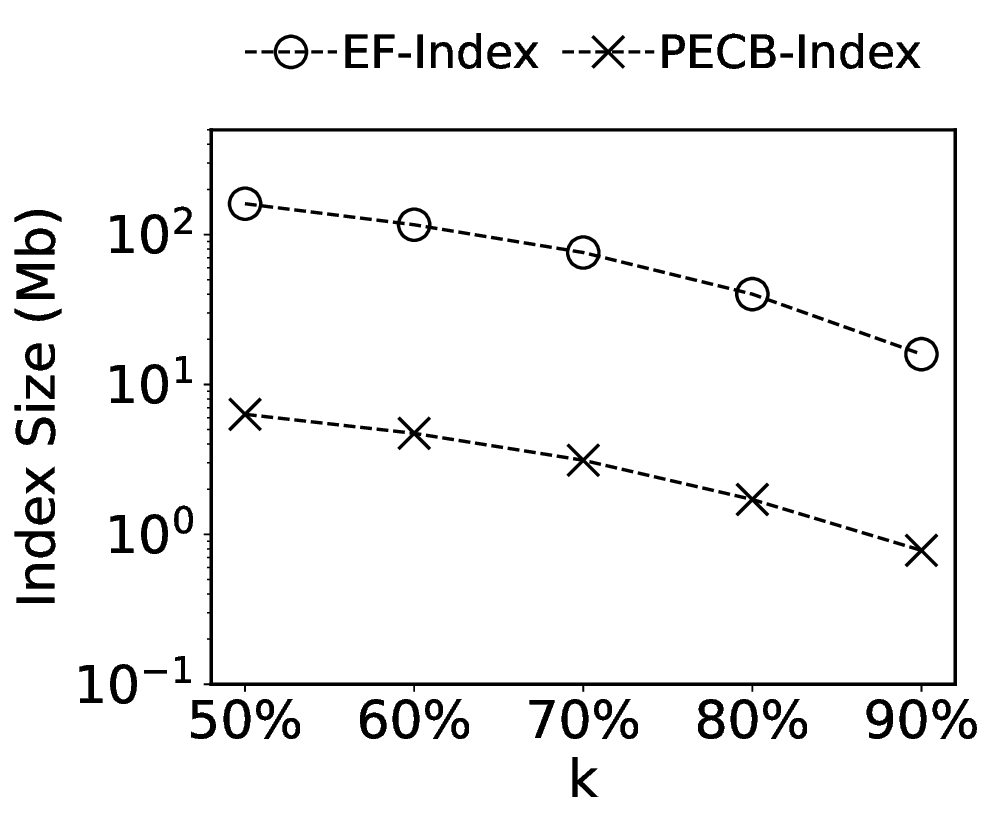}
    \caption{MathOverflow}
    \label{fig:em_vary_k_query}
  \end{subfigure}%
  \hfill
  \begin{subfigure}[b]{0.24\textwidth}
    \centering
    \includegraphics[width=\textwidth]{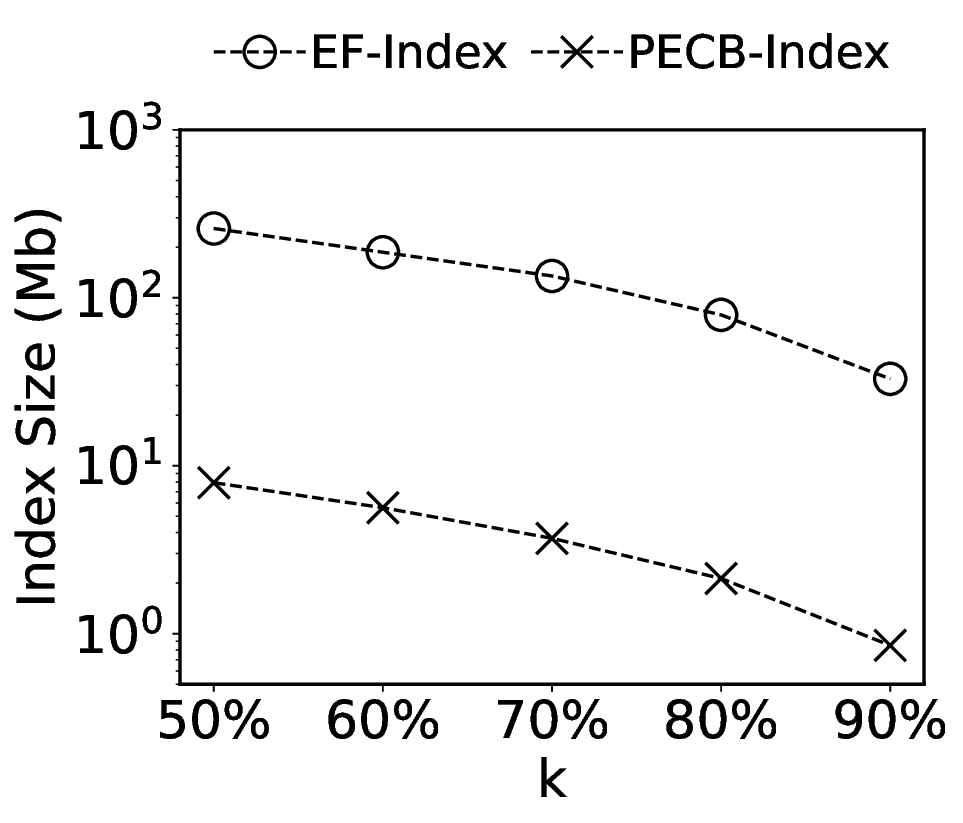}
    \caption{SuperUser}
    \label{fig:mf_vary_k_query}
  \end{subfigure}%
  \hfill
  \begin{subfigure}[b]{0.24\textwidth}
    \centering
    \includegraphics[width=\textwidth]{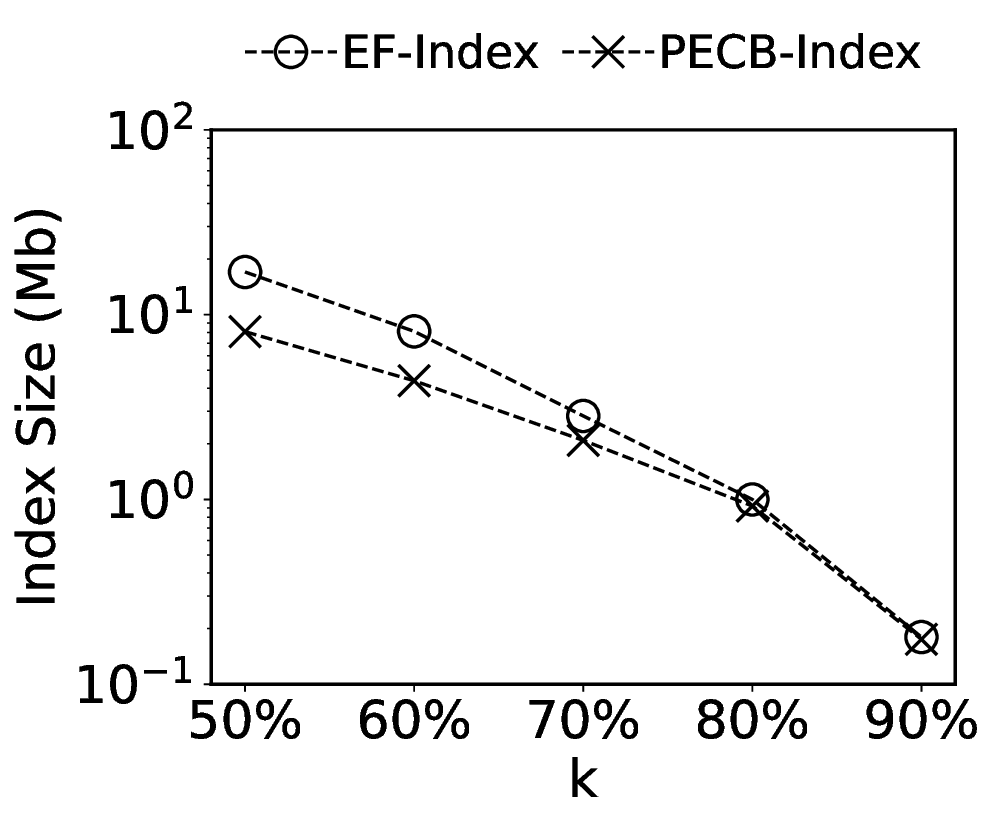}
    \caption{Youtube}
    \label{fig:yt_vary_k_query}
  \end{subfigure}
  \vspace{-1em}
  \caption{Index space for datasets with timestamp grouped by days, varying $k$ between 50\%, 60\%, 70\%, 80\%, and 90\% of $k_{max}$.}
  \label{fig:vary_k_space}
  \vspace{-0.5em}
\end{figure*}

\begin{figure*}[htbp]
  \centering

  \begin{subfigure}[b]{0.24\textwidth}
    \centering
    \includegraphics[width=\textwidth]{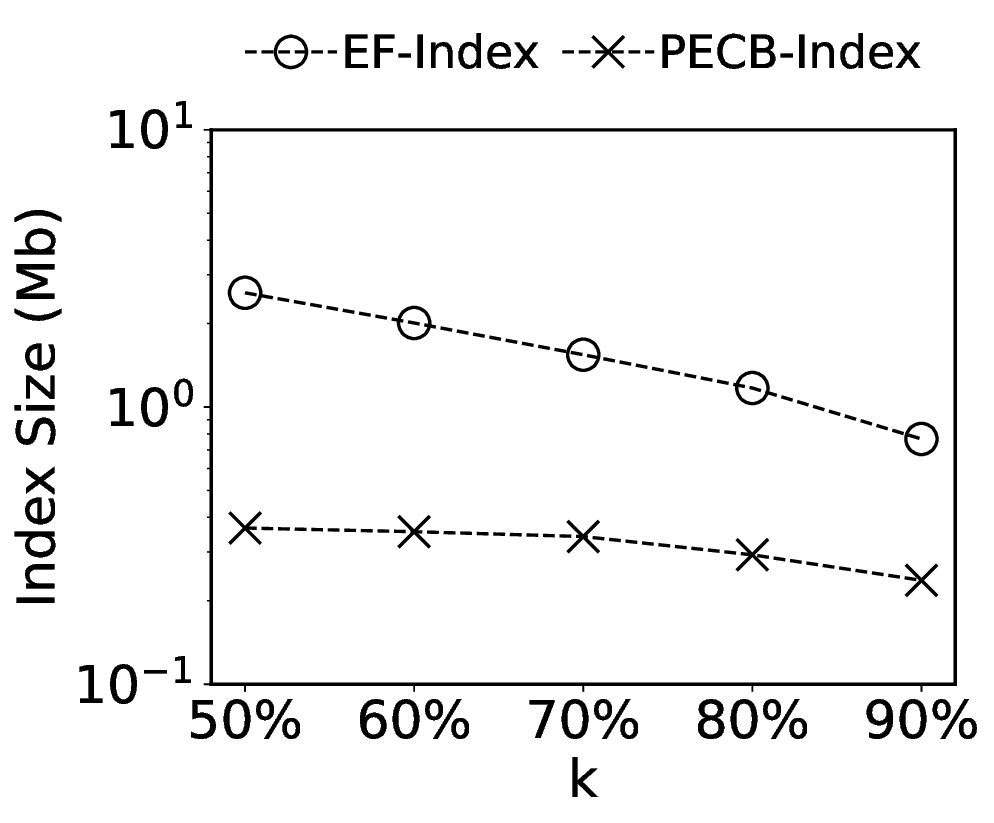}
    \caption{CollegeMsg}
    \label{fig:cm_vary_k_query}
  \end{subfigure}%
  \hfill
  \begin{subfigure}[b]{0.24\textwidth}
    \centering
    \includegraphics[width=\textwidth]{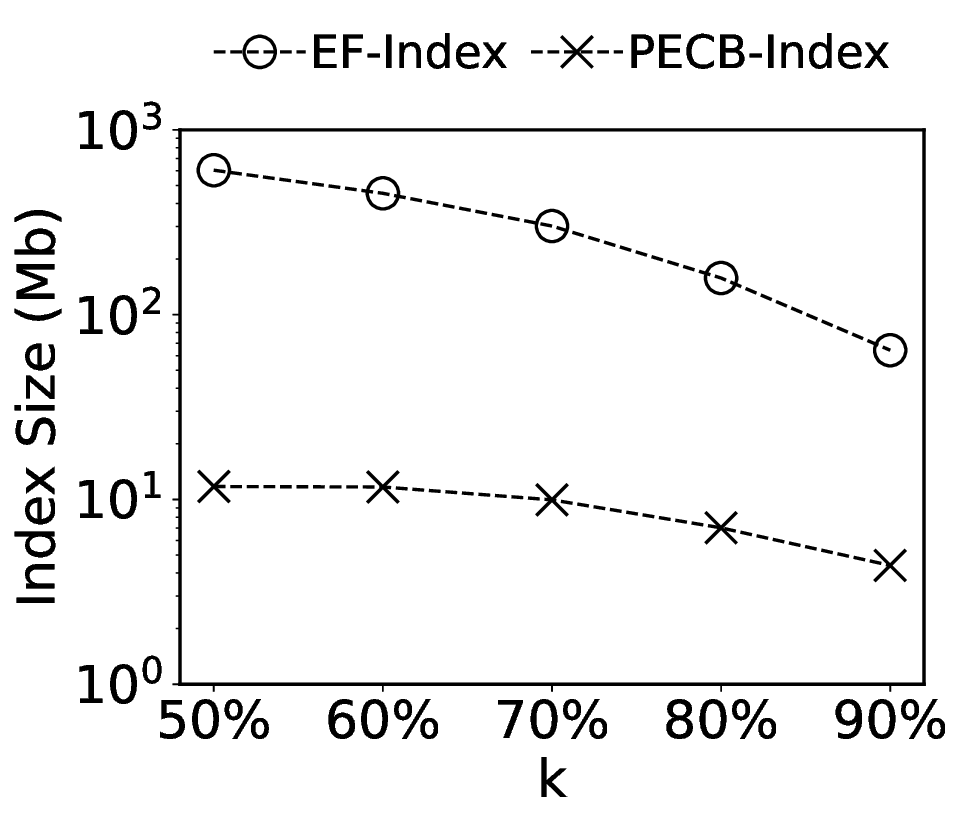}
    \caption{MathOverflow}
    \label{fig:em_vary_k_query}
  \end{subfigure}%
  \hfill
  \begin{subfigure}[b]{0.24\textwidth}
    \centering
    \includegraphics[width=\textwidth]{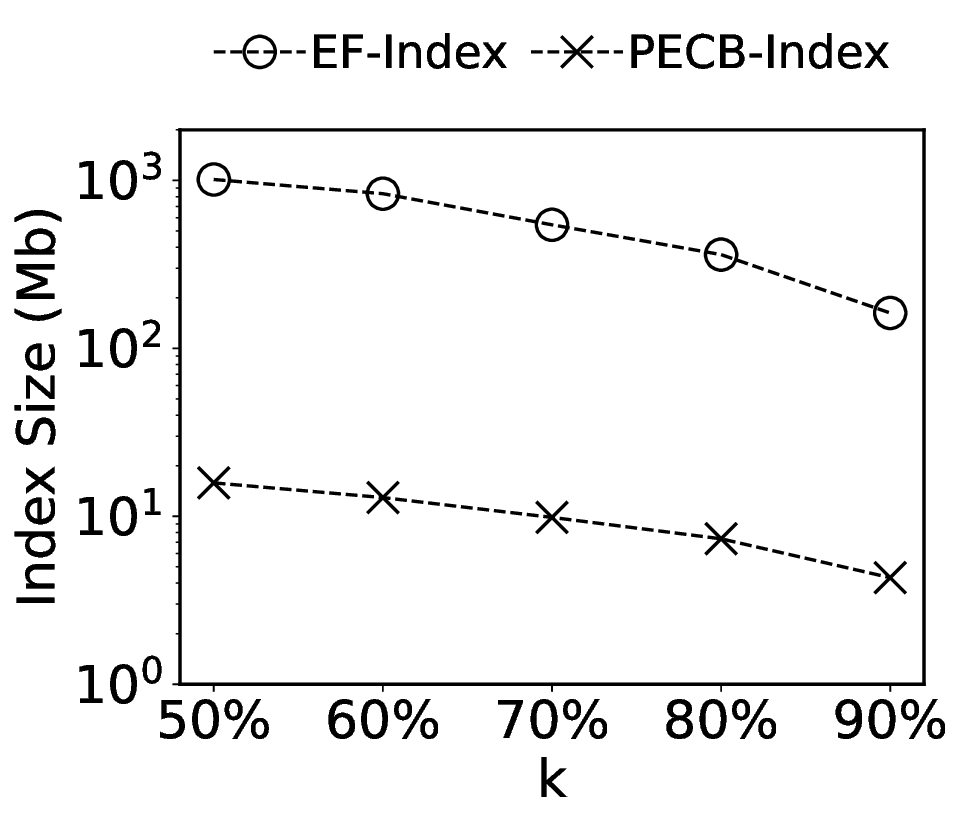}
    \caption{SuperUser}
    \label{fig:mf_vary_k_query}
  \end{subfigure}%
  \hfill
  \begin{subfigure}[b]{0.24\textwidth}
    \centering
    \includegraphics[width=\textwidth]{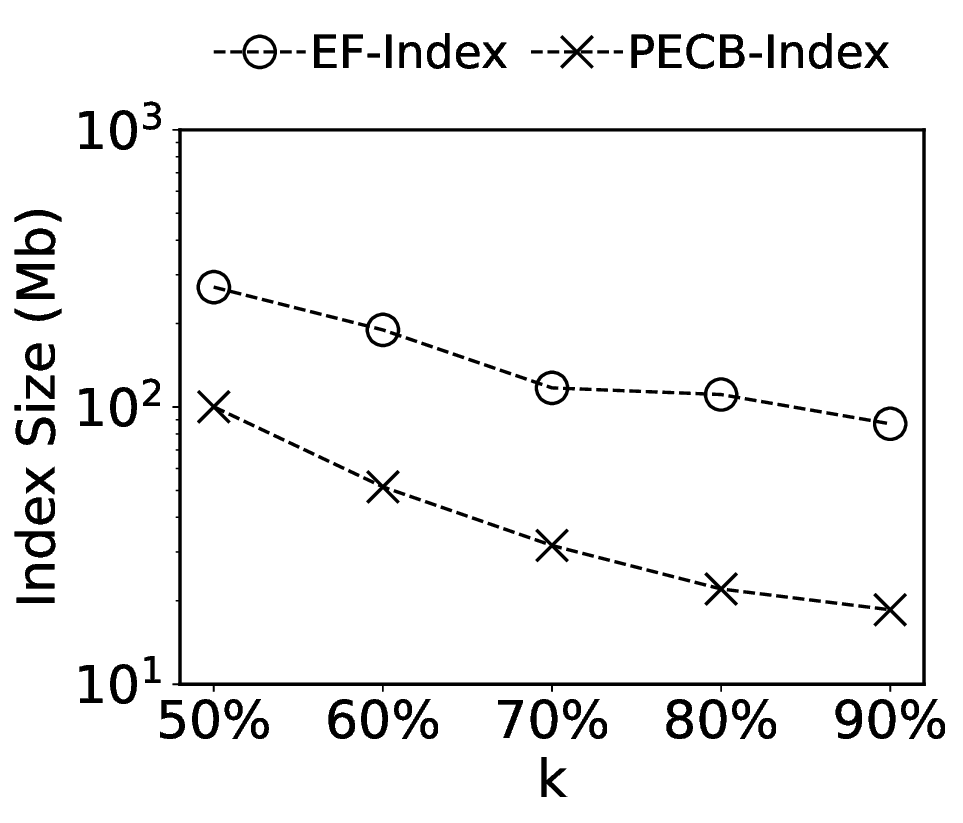}
    \caption{Youtube}
    \label{fig:yt_vary_k_query}
  \end{subfigure}
  \vspace{-1em}
  \caption{Construction time for datasets with timestamp grouped by days, varying $k$ between 50\%, 60\%, 70\%, 80\%, and 90\% of $k_{max}$.}
  \label{fig:vary_k_construction}
  \vspace{-0.5em}
\end{figure*}

\begin{figure*}[htbp]
  \centering

  \begin{subfigure}[b]{0.24\textwidth}
    \centering
    \includegraphics[width=\textwidth]{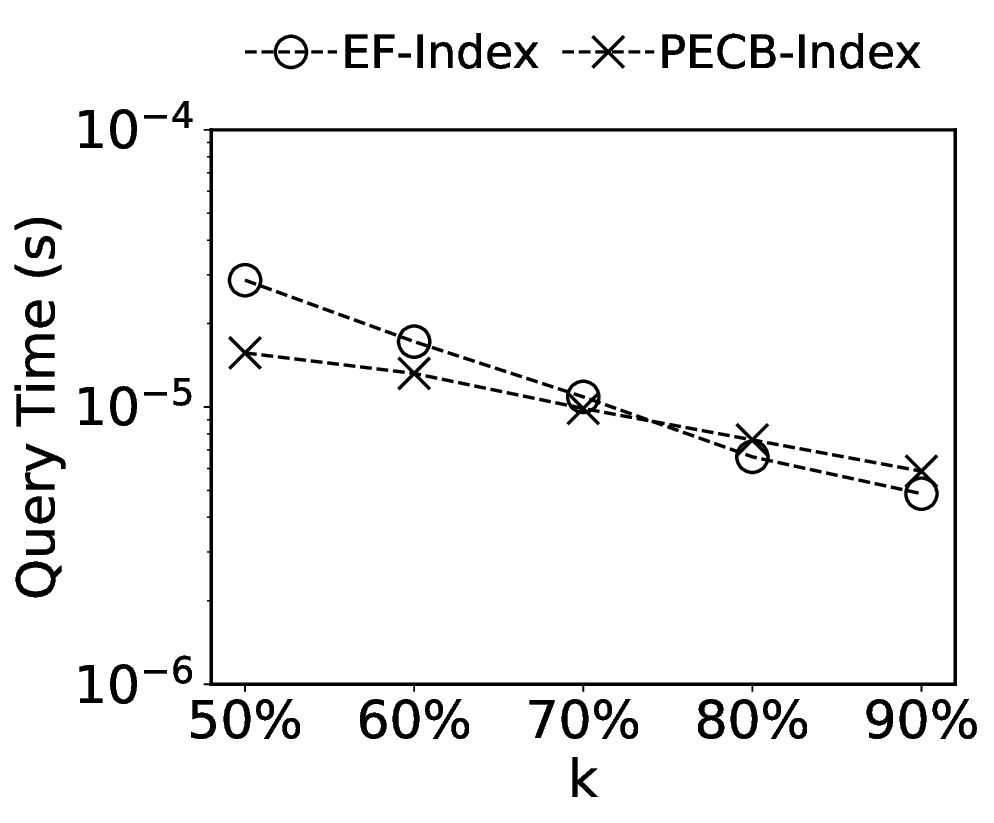}
    \caption{CollegeMsg}
    \label{fig:cm_vary_k_query}
  \end{subfigure}%
  \hfill
  \begin{subfigure}[b]{0.24\textwidth}
    \centering
    \includegraphics[width=\textwidth]{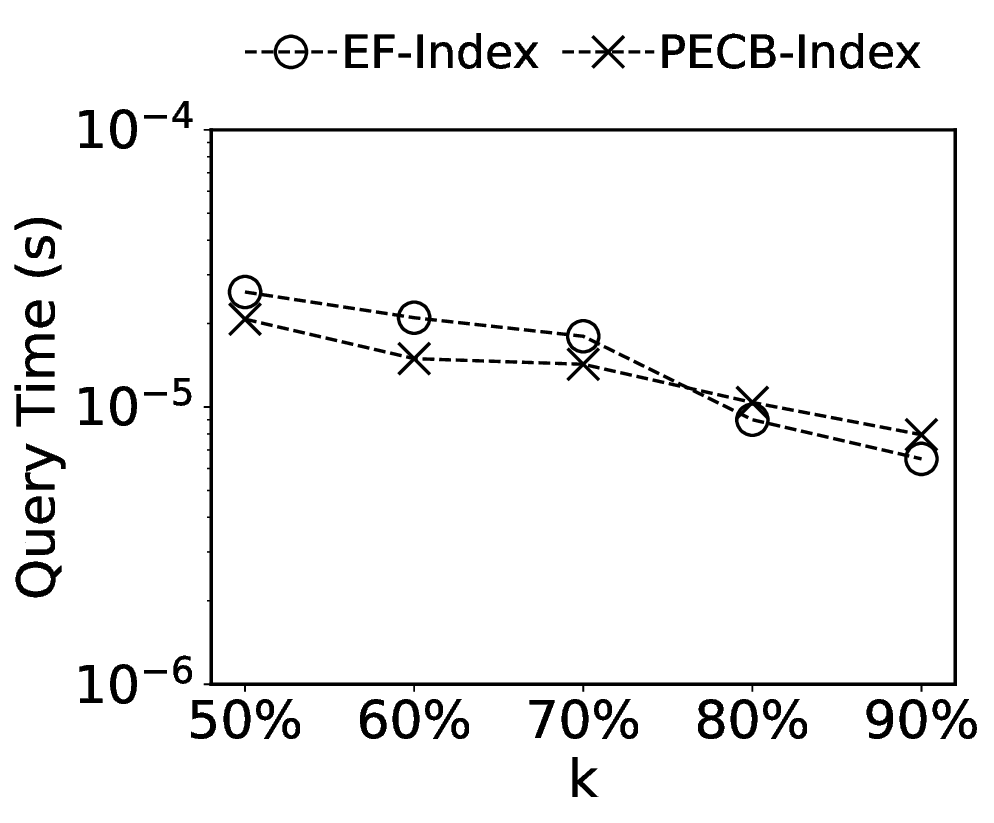}
    \caption{MathOverflow}
    \label{fig:em_vary_k_query}
  \end{subfigure}%
  \hfill
  \begin{subfigure}[b]{0.24\textwidth}
    \centering
    \includegraphics[width=\textwidth]{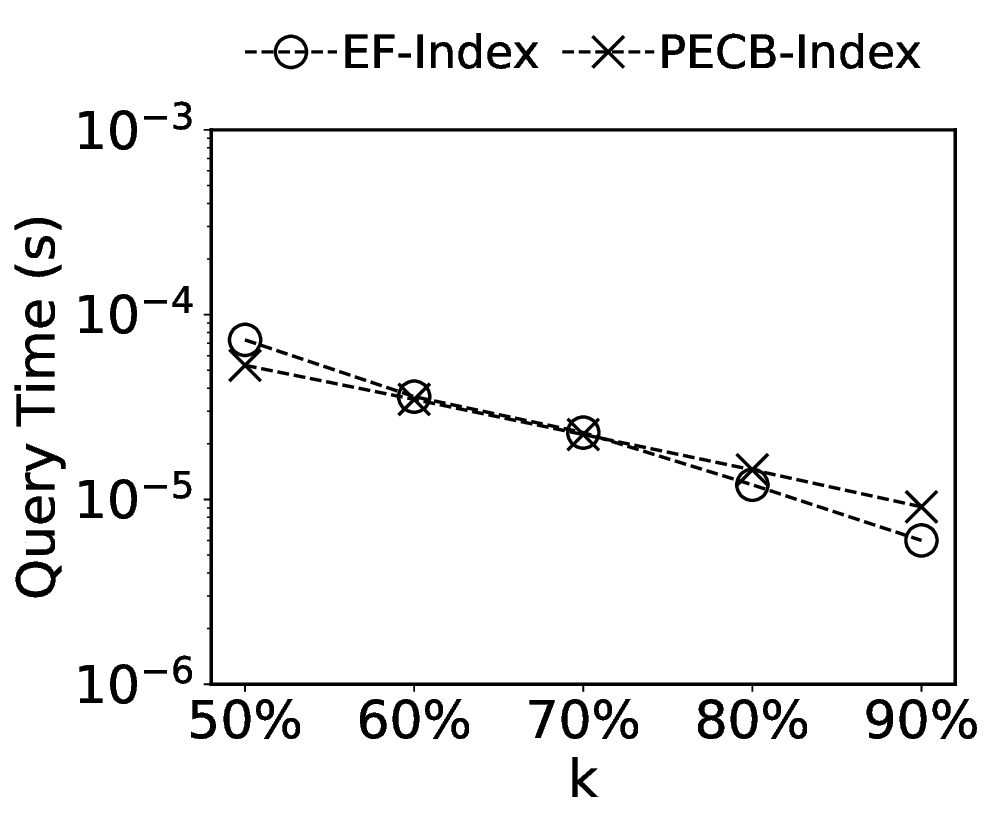}
    \caption{SuperUser}
    \label{fig:mf_vary_k_query}
  \end{subfigure}%
  \hfill
  \begin{subfigure}[b]{0.24\textwidth}
    \centering
    \includegraphics[width=\textwidth]{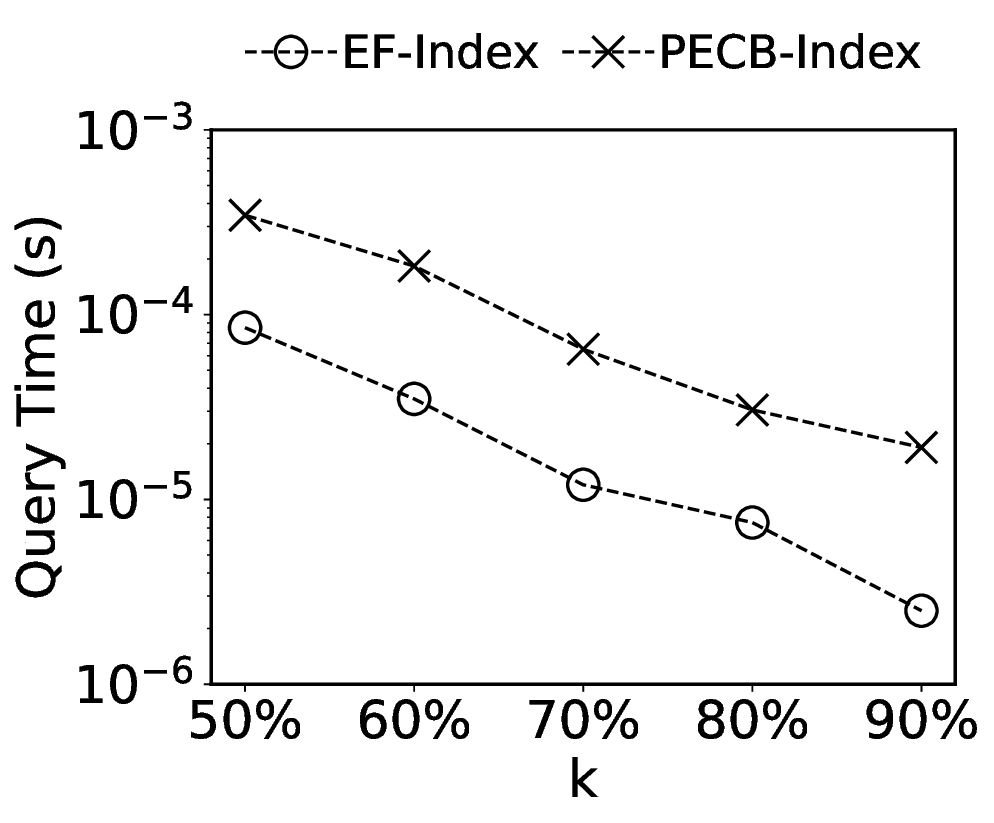}
    \caption{YouTube}
    \label{fig:yt_vary_k_query}
  \end{subfigure}
  \vspace{-1em}
  \caption{Average query time on datasets with timestamps grouped by days, varying \(k\) between 50\%, 60\%, 70\%, 80\%, and 90\% of \(k_{\max}\).}
  \label{fig:vary_k_query}
  \vspace{-0.5em}
\end{figure*}

\subsection{Index Storage Space}
\reffig{space_day} shows the space required by each index when timestamps are aggregated by day. Across every dataset, the PECB-Index uses between one and three orders of magnitude less space than the EF-Index. On the dataset with the most distinct days (PL), EF-Index occupies about 300 megabytes, CTMSF-Index uses roughly 40 megabytes, and PECB-Index remains under 20 megabytes—fifteen times smaller than EF-Index and twice as small as CTMSF-Index. On WK and MF, PECB-Index again shrinks EF-Index’s many-megabyte footprint down to single-digit megabytes and halves the size of CTMSF-Index. Even on graphs with very few days, such as FB, CM, or MC—where EF-Index already requires only a few tenths of a megabyte—PECB-Index still cuts space use in half compared to both baselines. Compared to CTMSF-Index, PECB-Index consistently uses two to four times less space by avoiding the long, frequently updated neighbor lists required at high-degree vertices: its edge-centric binary forest caps each node at two children and records neighbor changes only when they occur, bounding per-node storage without sacrificing correctness.

\subsection{Index Construction Time}
\reffig{construction_day} plots the index‐construction times under day‐aggregated timestamps. Our two incremental methods—CTMSF‐Index and PECB‐Index—are 1-3 orders of magnitude faster than EF‐Index, because they build a single, evolving structure instead of reconstructing an entire forest for every distinct TTI. On graphs with many distinct timestamps—such as MF, AU, LR, or WT—EF-Index spends hundreds or even thousands of seconds constructing hundreds of daily forests, while our incremental methods finish in just a few seconds, yielding speedups of up to three orders of magnitude. By contrast, on graphs like PL, YT, and DB, which have large edge counts but relatively few distinct days, the cost of EF-Index is lower and the gap narrows to around one order of magnitude. Finally, CTMSF-Index and PECB-Index exhibit almost identical build times, since they share the same core-time MSF construction; any slight advantage of PECB-Index reflects its more compact binary forest representation.

\begin{figure*}[htbp]
\centering
\includegraphics[width=\textwidth]{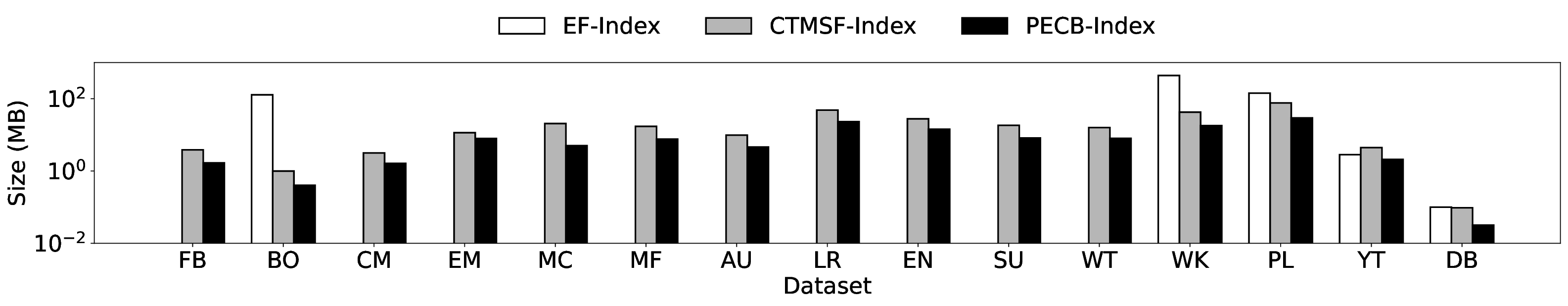}
\vspace{-2em}
\caption{Index space for the original datasets.}
\label{fig:space_org}
\vspace{-0.5em}
\end{figure*}

\begin{figure*}[htbp]
\centering
\includegraphics[width=\textwidth]{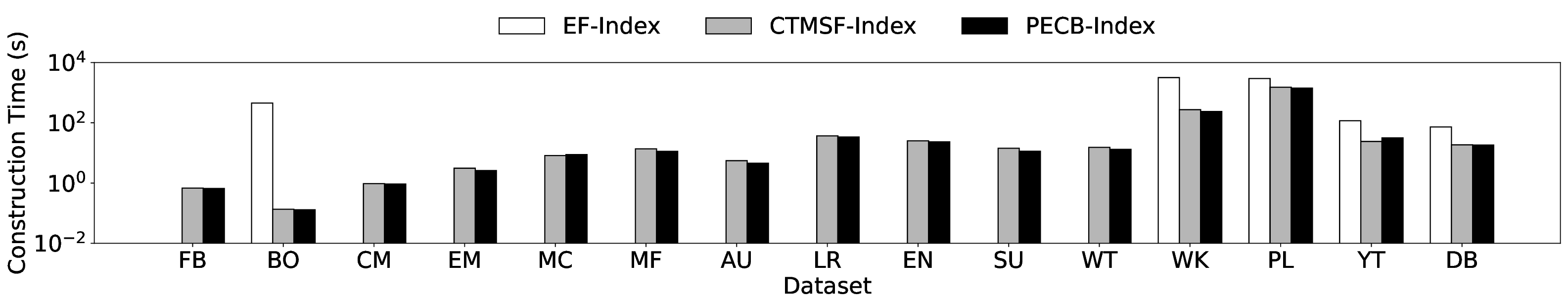}
\vspace{-2em}
\caption{Construction time for the original dataset.}
\label{fig:construction_org}
\vspace{-0.5em}
\end{figure*}

\begin{figure*}[htbp]
\centering
\includegraphics[width=\textwidth]{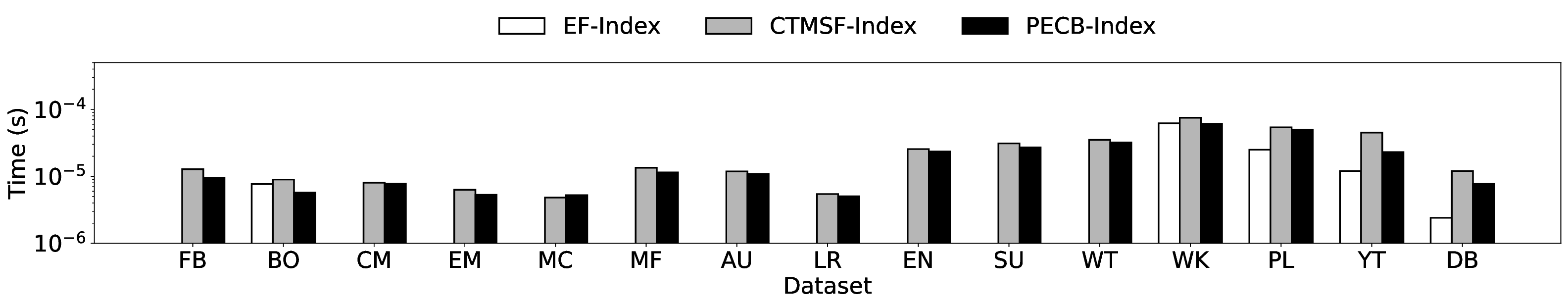}
\vspace{-2em}
\caption{Average query time on the original datasets.}
\label{fig:query_org}
\vspace{-0.5em}
\end{figure*}

\subsection{Query Time}
\reffig{query_day} shows the average query time of EF-Index, CTMSF-Index, and our PECB-Index on the day-aggregated datasets. PECB-Index achieves comparable query performance with EF-Index, with each approach offering distinct advantages. EF-Index maintains a separate evolution forest (MTSF) for every distinct tightest-time-interval (TTI) of a temporal $k$-core; once the correct forest structure is located, running a single BFS is inherently fast. Hence, on datasets with few days (e.g., MC or YT) and therefore few distinct TTIs, it outperforms our approach by roughly 2–5 times. However, as the number of distinct days grows, EF-Index must navigate through many prebuilt forests, and its per-query cost rises accordingly. In contrast, PECB-Index maintains one incrementally stored index and uses a fast binary search at each visited node to reconstruct the correct forest; this adds modest overhead—query times remain in the tens of microseconds—but scales gracefully with more days, often matching or even beating EF-Index on large-day datasets such as MF and EN. Finally, CTMSF-Index and PECB-Index exhibit nearly identical query performance, with PECB-Index enjoying a slight edge thanks to its binary-forest structure (each node has at most two children), which reduces the per-node search cost. Overall, these results confirm that PECB-Index achieves query speeds comparable to the state of the art while delivering vastly lower index size and faster build times.

\subsection{The Impact of $k$}
\reffig{vary_k_space} through \reffig{vary_k_query} explore how varying the cohesiveness parameter $k$ (from 50\% to 90\% of $k_{max}$) affects construction time, index space, and query time. As $k$ increases, every method becomes faster to query: larger values of $k$ yield smaller $k$-cores and therefore fewer nodes to traverse, with this trend most pronounced on coarse-grained datasets like YouTube, where the reduction in visited vertices is greatest. Likewise, index space steadily shrinks as $k$ grows, since only vertices in the $k$-core must be stored; again, YouTube shows the largest relative drop because its daily timestamp count is low. Construction time also falls with higher $k$, reflecting the smaller core-time forests to build; this effect is especially dramatic for EF-Index—which rebuilds a fresh forest for each interval—while PECB-Index benefits more modestly from its incremental, binary-forest updates. Across all $k$ settings, PECB-Index consistently uses less space and builds faster than EF-Index, demonstrating its robustness and efficiency regardless of the chosen cohesiveness level.

\subsection{Performance on Original Timestamps}
\reffig{space_org} through \reffig{query_org} present our experiments on the original, unaggregated timestamps. Here EF-Index cannot complete on most fine-grained datasets due to the sheer number of distinct timestamps. On the only datasets where EF-Index does finish (FB, CM, EM, MC, YT, and DB), both its index size and its construction time are higher than those of PECB-Index by 1-4 orders of magnitude. In contrast, PECB-Index (and CTMSF-Index) cope easily with high timestamp resolution by maintaining a single, incrementally updated structure and lightweight binary-forest operations. On PL, YT, and DB—where the raw timestamp count is actually lower than the day count—their performance matches the day-aggregated case. In short, PECB-Index succeeds on datasets with original timestamps where EF-Index cannot, while still dramatically outperforming it in space and build time.

\section{Conclusion}
\label{sec:conclu}

% In this work, we proposed an efficient framework for enumerating all temporal $k$-cores over a given time range in temporal graphs. By leveraging the concept of core times and introducing a minimal core window structure for edges, our temporal $k$-cores enumeration algorithm bounds the enumeration time by the size of the resulting temporal k-cores, which is optimal in practice. Extensive experimental evaluations demonstrated that our algorithm consistently outperforms existing methods by up to two orders of magnitude across diverse datasets.
%Additionally, we proposed the temporal $k$-core vertex set model, which offers the same informational value for $k$-core based cohesive sub-graph mining applications but with significantly reduced computational costs compared to existing temporal $k$-core models. 

We introduced an edge-centric binary forest that compactly records core-time changes and preserves all $k$-core components across windows. By storing only incremental differences and capping every node at two children, our method cuts construction time and space by up to three orders of magnitude compared with EF-Index while maintaining micro-second-level query latency. These results demonstrate that temporal $k$-core component search can be supported efficiently without the heavy redundancy of prior approaches.

% \begin{acks}
% \end{acks}

\clearpage
% \newpage

\bibliographystyle{ACM-Reference-Format}
\bibliography{ref}

\end{document}